\documentclass[12pt,letterpaper]{article}
\usepackage[top=1in, bottom=1in, left=1in, right=1in]{geometry}

\title{ When and what to learn in a changing world }
\author{ C\'esar Barilla }
\date{ \today }

\usepackage[utf8]{inputenc}
\usepackage[T1]{fontenc}
\usepackage[english]{babel}

\usepackage{layout}
\usepackage{setspace}
\usepackage{parskip}
\usepackage[bottom]{footmisc}     
\usepackage[all]{nowidow}         
\newcommand\blfootnote[1]{%
  \begingroup
  \renewcommand\thefootnote{}\footnote{#1}%
  \addtocounter{footnote}{-1}%
  \endgroup
}

\makeatletter
\let\thetitle\@title
\let\theauthor\@author
\let\thedate\@date
\makeatother

\usepackage{libertine} 

\usepackage{eurosym}
\usepackage{lipsum}

\usepackage{enumitem}

\usepackage{titlesec}
\titleformat*{\paragraph}{\scshape\bfseries}

\usepackage{mathtools} 
\usepackage{amssymb} 
\usepackage{amsmath}
\usepackage{amsfonts} 
\usepackage{amsthm}  
\usepackage{thmtools} 
\usepackage{dsfont}
\usepackage{stmaryrd}
\usepackage{mathtools} 
\usepackage{fourier}
\usepackage{amsbsy} 
\usepackage{bbm}  
\usepackage{units} 
\usepackage{cancel}

\DeclareMathAlphabet\mathbfcal{OMS}{cmsy}{b}{n}

\makeatletter
\def\thm@space@setup{%
  \thm@preskip=\parskip \thm@postskip=0pt
}
\makeatother

\newtheoremstyle{mainresultstyle}%
    {\topsep} 
    {\topsep} 
    {\itshape \setlength{\parindent}{1.5em}} 
    {.5em}
    {\bfseries}
    {.}
    {0.5em}
    {}

\theoremstyle{mainresultstyle}
  \newtheorem{theorem}{Theorem}
  \newtheorem{proposition}{Proposition}
  
  \newtheorem{appendixproposition}{Proposition}[section]
  \newtheorem{lemma}{Lemma}
  \newtheorem{appendixlemma}{Lemma}[section]
  
  \newtheorem{appendixcorollary}{Corollary}[section]

\newtheoremstyle{annexresultstyle}%
    {\topsep} 
    {\topsep} 
    {\itshape \setlength{\parindent}{1.5em}} 
    {.5em}
    {\bfseries}
    {.}
    {0.5em}
    {}

\theoremstyle{definition}
  \newtheorem{definition}{Definition}
  \newtheorem{appendixdefinition}{Definition}[section]

\newtheoremstyle{assumptionthmstyle}%
    {\topsep} 
    {\topsep} 
    {} 
    {.5em}
    {\scshape}
    {.}
    {0.5em}
    {}

\theoremstyle{assumptionthmstyle}

\theoremstyle{remark}

\DeclareMathOperator*{\argmin}{arg\,min}
\DeclareMathOperator*{\argmax}{arg\,max}

\DeclareMathOperator{\interior}{int}
\DeclareMathOperator{\supp}{supp}

\newcommand{\prob}{\mathbb{P}}
\newcommand{\reals}{\mathbb{R}}
\newcommand{\Esp}{\mathbb{E}}
\newcommand{\Var}{\mathbb{V}}

\newcommand{\naturals}{\mathbb{N}}

\makeatletter
\providecommand{\leftsquigarrow}{%
  \mathrel{\mathpalette\reflect@squig\relax}%
}
\newcommand{\reflect@squig}[2]{%
  \reflectbox{$\m@th#1\rightsquigarrow$}%
}
\makeatother

\allowdisplaybreaks

\usepackage{graphicx} 
\usepackage{wrapfig} 
\usepackage[font=small, labelfont=sc]{caption} 
\usepackage{subcaption}
\usepackage{pdfpages}
\usepackage{pgfplots}
\usepackage{tikz}
\usetikzlibrary{arrows,automata,shapes,patterns,decorations.pathreplacing,calligraphy,shapes.arrows}

\usepackage{environ}
\usepackage{etoolbox}

\usepackage{xcolor}
\definecolor{myblue}{rgb}{.1,.3,1}
\definecolor{mygreen}{rgb}{.1,.7,.2}
\definecolor{myred}{rgb}{1,.1,.2}
\definecolor{myorange}{RGB}{230,150,50}

\definecolor{pyplotorange}{rgb}{1,.498,.055}

\definecolor{mathematicablue}{rgb}{0.368417, 0.506779, 0.709798}
\definecolor{mathematicaorange}{rgb}{0.880722, 0.611041, 0.142051}
\definecolor{mathematicagreen}{rgb}{0.560181, 0.691569, 0.194885}
\definecolor{mathematica4}{rgb}{0.922526, 0.385626, 0.209179}
\definecolor{mathematica5}{rgb}{0.528488, 0.470624, 0.701351}

\usepackage[authoryear,round]{natbib}

\usepackage[hypertexnames=false,backref=page]{hyperref} 
\newcommand*{\hyref}[2]{\hyperref[#2]{#1~\ref{#2}}}
\hypersetup{
    pdffitwindow=false,                 
    pdfstartview={XYZ null null 1.00},          
    pdfnewwindow=true,                  
    colorlinks=true,                  
    linkcolor={red!50!black},             
    citecolor={blue!50!black},              
    urlcolor={red!50!black},                   
  linkbordercolor={white},              
}

\makeatletter
\let\@msm@th@eqref\eqref
\renewcommand{\eqref}[1]{%
  \begingroup
  \leavevmode
  \color{red!50!black}%
  \hypersetup{linkbordercolor=[named]{red!50!black}}%
  \@msm@th@eqref{#1}%
  \endgroup
}
\makeatother

\newtoggle{displaysectionsummary}      
\newtoggle{displaysectioncontent}      
\newtoggle{displayoutlinenotes}      
\newtoggle{displaygaps}
\newtoggle{displaynotes}
\newtoggle{displayissues}

\newlist{sectioncontentlist}{itemize}{1}
\setlist[sectioncontentlist,1]{label =, itemindent= 2em}

\usepackage[most]{tcolorbox}

\newtcolorbox{sectionsummarybox}{
  enhanced, 
  breakable,
  title={Contents: },
  coltitle=gray,
  fonttitle=\scshape \bfseries,
  attach title to upper = {\ \ },
  coltext=gray,
  colback=white,
  colbacktitle=white,
  colframe=gray!50!white,
  boxrule=.5pt,
  boxsep=.5pt,
  width=.9\textwidth,
  sharp corners=north,
  center,
  before=\ \par\smallskip,
}

\NewEnviron{sectionsummary}{%
  \iftoggle{displaysectionsummary}{
    \begin{sectionsummarybox}
      \BODY
    \end{sectionsummarybox}
  }{}%
}

\NewEnviron{sectioncontent}{%
  \iftoggle{displaysectioncontent}{\BODY}{}%
}

\AtBeginDocument{%
  \setlength\abovedisplayskip{8pt}
  \setlength\belowdisplayskip{7pt}
  \togglefalse{displaysectionsummary}
  \toggletrue{displayoutlinenotes}
  \toggletrue{displaysectioncontent}
  \toggletrue{displaygaps}
  \toggletrue{displaynotes}
  \toggletrue{displayissues}
  }

\DeclareMathOperator{\BellmanOperator}{\Phi}
\DeclareMathOperator{\InfoVal}{\mathbfcal{G}} 
\DeclareMathOperator{\CavDist}{\Gamma}
\DeclareMathOperator{\StopVal}{\mathbfcal{S}} 
\DeclareMathOperator{\Cprocess}{\mathfrak{C}} 
\DeclareMathOperator{\Upath}{\mathcal{U}} 
\DeclareMathOperator{\Uprocess}{\mathfrak{U}} 
\newcommand{\id}{i} 
\DeclareMathOperator{\inforegion}{\mathcal{I}}

\DeclareMathOperator{\grossvalregion}{\mathcal{E}}
\DeclareMathOperator{\inforegionopt}{\inforegion^*}

\DeclareMathOperator{\grossvalregionopt}{\grossvalregion^*}
\DeclareMathOperator{\Cav}{Cav} 
\DeclareMathOperator{\CandidateValue}{\mathbb{V}}
\DeclareMathOperator{\BayesPlausible}{\mathcal{B}}
\DeclareMathOperator{\beliefprocesses}{\mathbfcal{B}}
\DeclareMathOperator{\beliefprocessesdiscrete}{\mathbfcal{B}_d}
\DeclareMathOperator{\beliefcycle}{\Upsilon}
\DeclareMathOperator{\piint}{int_\pi}
\DeclareMathOperator{\inboundary}{\partial^{\text{in}}_\pi}
\DeclareMathOperator{\outboundary}{\partial^{\text{out}}_\pi}

\DeclareMathOperator{\waitorconfirmprocess}{WoC}

\newcommand{\SymmetricalCyclePayoffs}{v^S}
\newcommand{\lowtarget}{q^0}
\newcommand{\hightarget}{q^1}
\newcommand{\lowthreshold}{p^0}
\newcommand{\highthreshold}{p^1}
\newcommand{\lowtime}{\tau^0}
\newcommand{\hightime}{\tau^1}
\newcommand{\targetsym}{q}
\newcommand{\thresholdsym}{p}
\newcommand{\freqsym}{\tau}

\graphicspath{{./Plots}}

\begin{document}

\setcounter{footnote}{0}
\renewcommand{\thefootnote}{\fnsymbol{footnote}}
\vspace*{3em}
\begin{center}\huge
    {\noindent
    \bfseries \thetitle}
\end{center}
\vspace*{1em}

\makebox[\textwidth][c]{
\begin{minipage}{1.2\linewidth}
\Large\centering
\theauthor\footnotemark
\end{minipage}
}
\setcounter{footnote}{1}\footnotetext{\setstretch{1}
Department of Economics, Oxford University and Nuffield College;
cesar.barilla@economics.ox.ac.uk}

\blfootnote{
    I am indebted to Yeon-Koo Che, Navin Kartik, Laura Doval, and Elliot Lipnowski for their guidance and support. 
    I am grateful to 
        Hassan Afrouzi,
        Arslan Ali,  
        Mark Dean, 
        Alkis Georgiadis-Harris, 
        Duarte Gonçalves, 
        Jan Knoepfle, 
        Jacopo Perego, 
        Bernard Salanié, 
        José Scheinkman, 
        Ludvig Sinander, 
        Akanksha Vardani, 
        Nikhil Vellodi, 
        Yu Fu Wong,
        and 
        Mike Woodford 
    for their comments.
    I thank audiences at 
        Columbia, 
        Oxford, 
        EAYE '24, 
        IO Theory '24, 
        the Paris Game Theory Seminar, 
        CETC '25, 
        SAET '25, 
        EC '25, 
        Stony Brook IGTC '25,
        ESWC '25,
        EEA congress '25,
        the 2025 Transatlantic Theory Workshop,
        the Paris School of Economics,
        Essex,
        Warwick,
        and Helsinki GSE
        for feedback.
}

\setcounter{footnote}{0}
\renewcommand{\thefootnote}{\arabic{footnote}}
\vspace*{-4em}

\begin{center} 
\small
This version: \today
\end{center}

\vspace{20pt}

\begin{center}
\textbf{\scshape Abstract}
\vspace{.5em}
\end{center}
\noindent\makebox[\textwidth][c]{
    \begin{minipage}{.8\textwidth}
        \noindent
        A decision-maker periodically acquires information about a changing state, controlling both the timing and content of updates. 
        I characterize optimal policies using a decomposition of the dynamic problem into optimal stopping and static information acquisition.
        Eventually, information acquisition either stops or follows a simple cycle in which updates occur at regular intervals to restore prescribed levels of relative certainty. 
        This enables precise analysis of long run dynamics across environments.
        As fixed costs of information vanish, 
        belief changes become lumpy: it is optimal to either wait or acquire information so as to exactly confirm the current belief until rare news prompts a sudden change.
        The long run solution admits a closed-form characterization in terms of the "virtual flow payoff".
        I highlight an illustrative application to portfolio diversification.
        ~
    \end{minipage}
}

\newpage

\section{Introduction}\label{sec:Intro}

    The world is constantly changing. 
    Yet, most decision makers are not continuously upending their worldview. 
    When making frequent decisions, it seems reasonable in the short run to act based on previously held beliefs or as if the relevant conditions are approximately fixed.
    However, past information eventually becomes outdated.
    Periodically seeking to improve knowledge of current circumstances may be profitable, even if it is costly.
    This raises two natural questions: \emph{when} should one decide to acquire new information and \emph{what} should they learn?

    Consider for instance the problem of an investor allocating resources between assets with uncertain returns. 
    Market trends and fundamentals that govern asset performance may change over time.
    How often and how thoroughly should the investor reconsider their current views? 
    They could opt for infrequent but detailed research or frequent, less precise monitoring.
    The optimal balance between timing and quality of information acquisition depends on the stakes, information acquisition costs, and underlying volatility.  
    In volatile markets, there is more to learn, but information becomes outdated faster. 
    Higher certainty may be required for riskier investments, leading to quicker depreciation of information and higher costs.
    The investor's example is representative of a large class of problems: a government splitting budget between agencies with evolving needs, a producer choosing between available technologies, a retailer allocating inventory between locations with fluctuating demands, among others.

    In this paper, I study a dynamic model of optimal information acquisition about a changing world, which provides a rich yet tractable way to capture the relation between the timing and content of infrequent information acquisition.
    A decision maker (DM) takes an action repeatedly at every instant; flow payoffs depend on their action choice and on an unobserved binary state of the world which changes over time. 
    The DM sequentially chooses times at which they wish to acquire some information, which entails a fixed cost. 
    At each such time they also flexibly decide what to learn, which entails a variable cost.
    The model is introduced in \hyref{Section}{sec:Model}.

    The first contribution of the paper is to rigorously solve the problem (\hyref{Section}{sec:Characterization-Of-Solutions}).
    Well-known difficulties arise from the recursive nature of the value of information: the incentives to acquire information today simultaneously depend on all future expected information and how the state changes; as a result, potentially complex learning dynamics induce nonlinear continuation values.
    However, the combination of the continuous time structure with infrequent information acquisition makes this model tractable. 
    I derive an appropriate Bellman equation that decomposes the DM's problem into a static optimal information acquisition problem and an optimal stopping problem. 
    I show that the value function uniquely solves this equation even though the Bellman operator is not a contraction (\hyref{Theorem}{thm:recursive-characterization}).
    Optimal policies must consistently combine properties derived from optimal stopping and static information acquisition.
    This enables the precise characterization of optimal information acquisition (\hyref{Theorem}{thm:optimal-policies}), which is described by two nested collections of belief intervals.

    Second, I study the induced dynamics of information acquisition (\hyref{Section}{sec:Dynamics}).
    The main result is that optimal information acquisition must eventually either stop or settle into a simple repeating cycle (\hyref{Theorem}{thm:long-run-dynamics}).
    In the cyclical case, information is acquired at regular intervals of time, when uncertainty reaches specific thresholds. 
    Updates lead to two possible outcomes, captured by two "target posterior beliefs" reflecting endogenously chosen levels of relative confidence that one state is more likely than average. 
    Each possible outcome leads to a waiting period of fixed length until the next update. 
    In practice, this rules out more complex strategies with intermediary or irregular updates.
    This further simplifies the long run dynamics of optimal information acquisition: the problem reduces to choosing the \emph{content} and \emph{frequency} of updates (\hyref{Proposition}{prop:directly-stationary-problem}); resulting expected payoffs have closed form expressions.

    The convergence result enables precise characterizations of long run behavior, which neatly captures dynamic incentives from repeated information acquisition and enriches static insights.
    Optimal information acquisition may exhibit path dependence in the form of "learning traps" (\hyref{Proposition}{prop:when-does-learning-stop}): if initial beliefs belong to a "trap region" of sufficiently uninformed beliefs, then no learning ever occurs even though information would be perpetually acquired with better initial information.
    Comparisons of information acquired under different policies or environments is a theoretically challenging question for general dynamic processes, but cyclical dynamics outline three distinct types of "long-run informativeness" (\hyref{Definition}{def:comparison-of-belief-cycles}).
    For instance, I show that the world becoming more volatile has non-monotonic effects on the frequency of information acquisition (\hyref{Proposition}{prop:CS-lambda-tau}). 
       
    Third, I study the limit of the model when fixed costs vanish and derive an explicit characterization for optimal information acquisition in the limit (\hyref{Section}{sec:Vanishing-Fixed-Costs}). 
    Without fixed costs, it is optimal for the DM to either wait or acquire infinitesimal amounts of information to exactly confirm their current belief until rare news prompts a jump to a fixed alternative belief; furthermore optimal policies must converge to a policy of this kind as fixed costs vanish (\hyref{Theorem}{thm:optimality-and-convergence-to-wait-or-confirm}).
    In the long run, only two possible beliefs are ever held as the DM exactly prevents depreciation of information. 
    The long run optimal belief process admits a closed form characterization which derives from the concavification of an appropriately defined "virtual flow payoff" function (\hyref{Theorem}{thm:explicit-stationary-policy-no-fixed-cost}).
    The concentration on two beliefs delivers a new potential resolution of a classical tension between the fact that empirical decision makers often change actions lumpily while most models of dynamic learning predict continuous adjustment if beliefs are continuously changing and different beliefs imply different actions.
    
    Lastly, I provide an illustrative application of the framework to a portfolio allocation problem (\hyref{Section}{sec:Examples-Applications}). 
    Optimal behavior exhibits continuous rebalancing of the portfolio towards more diversification, punctuated by periodic shifts to a more extreme allocation.
    There may be information traps where initially uninformed investors are never able to acquire information and only ever buy a safe asset while informed ones retain better information and higher returns from risky assets.
    If fixed costs of information acquisition are negligible, information traps and continuous rebalancing disappear: investors always hold risky portfolios, adjusted only at discrete points in time.
    The frequency of adjustments is proportional to the underlying volatility of the environment. 
    In some cases with asymmetries between assets, optimal information acquisition can generate distortions between responsiveness to "good" and "bad" news, which connects to stylized facts from the literature on financial attention.

    \subsection*{Related Literature}\label{sec:Related-Lit}

        Understanding imperfect adjustments to changing conditions is a long standing theoretical and empirical agenda, with informational approaches gaining more attention in recent literature.
        \citet{mankiw2002sticky,reis2006inattentiveproducers,reis2006inattentiveconsumers} model "inattentiveness" via a fixed observation cost.
        \citet{sims2003implications} instead considers limited capacity in flexible acquisition of information, using entropy reduction.
        Subsequent literature on dynamic rational inattention (DRI) largely focuses on environment with quadratic payoffs and gaussian states and information \citep[e.g.][]{mackowiak2009optimal,mackowiak2018dynamic,afrouzi2021dynamic,davies2024learning}, or on random choice implications of Shannon costs \citep{steiner2017rational}.
        \citet{khaw2017discrete} experimentally test discrete adjustment models and provide evidence for both lumpy and flexible attention.
       
        Adjustment to a changing world also arises in experimentation problems -- see \citet{whittle1988restless} for a seminal reference and \citet{che2024predictive} for a recent contribution. 
        Unlike with costly information acquisition, information in those problems is entangled with action choices; belief dynamics may look similar, but drivers and predictions differ because the agent can learn about alternatives without changing the current action.
        In social learning with a changing state \citep[e.g.][]{moscarini1998social,dasaratha2023learning}, learning occurs once from past actions and dynamics are equilibrium-driven rather than from forward-looking optimization; \citet{levy2022stationary} study a related steady-state environment with costly information acquisition.
        
        The leading application is information acquisition in finance. 
        \citet{van2010information} study strategic information acquisition and portfolio diversification in a static setting.
        “Ostrich effects”—more monitoring after good than bad news—are documented empirically \citep[e.g.][]{karlsson2009ostrich,galai2006ostrich,sicherman2016financial}.
        Periodic inspections also appear in monitoring of strategic agents \cite[e.g.][]{varas2020random,wong2023dynamic,Ball_2023}, where optimality is driven by incentive compatibility rather than informational motives.

        Recent work on dynamic information acquisition \citep{che2019optimal,zhong2022optimal,hebert2023rational,georgiadisharris2023} features flexible continuous acquisition with persistent states and a single decision; incentives there stem from convex (or budget) costs over information flow, whereas I assume linear flow costs so dynamics are driven by the changing state.
        This builds on static information acquisition and design: posterior separable costs \citep{caplin2022rationally,denti2022posterior} generalize Shannon costs \citep{sims2003implications}; concavification dates to \citet{aumann1995repeated} and underlies Bayesian persuasion with and without costs \citep{kamenica2011bayesian,gentzkow2014costly}; see \citet{ely2017beeps} for dynamic persuasion.

\section{Model}\label{sec:Model}


    \paragraph{Environment and decision problem} 
    Time is continuous and indexed by $t \geq 0$. 
    There is a single decision maker (DM) who takes an action $a_t \in A$ at every instant in time; this generates flow payoffs which depend on the current action choice and the current value of a binary state of the world $\theta_t \in \Theta := \{0,1\}$.
    Denote by $\tilde{u}: A \times \Theta \rightarrow \reals$ the utility function mapping actions and states to payoffs. Both the state and flow payoffs are unobserved.

    The decision problem induces \emph{indirect utility function} $u$, which maps beliefs about the current state to expected payoff from the optimal action choice under those beliefs. 
    Formally, denoting $\Delta(\Theta)$ the space of probability distributions over $\Theta$, $u:\Delta(\Theta) \rightarrow \reals$  is defined as:
    \begin{align*}
    u(p) := \max_{a \in A} \Esp_{\theta \sim p} \bigl[ \tilde{u}(a,\theta) \bigr].
    \end{align*}
    Assume an optimal action exists and $u$ is continuous (e.g. $A$ compact and $\tilde{u}$ continuous).
    The model and results are unchanged if one considers instead an arbitrary continuous indirect utility function $u(p)$, which may capture interactions with other strategic agents in reduced form -- see for instance \hyref{Section}{sec:Discussion-Extensions} for a dynamic persuasion interpretation.
    
    \paragraph{State transitions} 
    The state $\theta_t$ changes stochastically over time and follows Markovian dynamics: it jumps from $0$ to $1$ at rate $\lambda_0 > 0$ and from $1$ to $0$ at rate $\lambda_1 > 0$.
    Given that the state space is binary, the space of beliefs over $\Theta$ can be identified with the unit interval $[0,1]$, labeling beliefs in terms of the probability of the current state being $1$.
    Markovian dynamics can be conveniently reparameterized in terms of the \textbf{total transition rate} $\lambda > 0$ and \textbf{invariant distribution} $\pi \in (0,1)$, which are defined as:
    \begin{align*}
        \lambda & := \lambda_0 + \lambda_1, \vspace{.1em} \text{ and \vspace{.2em} } \pi := \frac{\lambda_0}{\lambda_0+\lambda_1}.
    \end{align*}
    Intuitively, $\pi$ captures the long run average proportion of time that the state spends at $1$; $\lambda$ captures the total rate at which the state changes, which I will interpret as overall volatility.
    
    \paragraph{Information acquisition and beliefs}  
    The DM chooses \emph{when} to acquire information, and \emph{what} information to acquire whenever they do. An information acquisition policy is described by sequences of (random) information acquisition times and information structures $\{\tau_i,F_i\}_{i \in \naturals}$ contingent on past information, where:
    \begin{itemize}
        \item $\{\tau_\id\}_{\id \in \naturals}$ are \textbf{information acquisition times}, i.e. $\tau_\id \in \overline{\reals}_+$ is the $\id$-th time of information acquisition. The $\tau_i$ are a.s. increasing, strictly so when finite.
        \item $\{F_\id\}_{\id \in \naturals}$ are \textbf{information structures}, i.e. the content of signals being acquired at each $\tau_\id$. As is now standard in the information acquisition literature, each information structure is represented as a \emph{probability distribution over posterior beliefs}: $F_i \in \Delta \Delta(\Theta)$ for all $i$.
    \end{itemize}

    The information acquisition policy $\{\tau_\id,F_\id\}_{\id \in \naturals}$ induces the belief process $\{P_t\}_{t \geq 0}$ as follows. 
    In between moments of information acquisition, beliefs about the current state drift towards the long run average $\pi$ at an exponential rate controlled by $\lambda$: even in the absence of new information a Bayesian agent is aware that the hidden state might have changed.
    Fix some initial belief $p$ and normalize the current time to $0$; until the next update beliefs evolve according to:
    \begin{align*}
        dp_t = \lambda(\pi-p_t) dt \text{, \hspace{.3em} or equivalently: } p_t = e^{- \lambda t} p + \bigl( 1-e^{-\lambda t} \bigr) \pi.
    \end{align*}
    Throughout the paper, I use lowercase $p_t$ to denote the deterministic path of beliefs starting from $p_0 = p \in \Delta(\Theta)$ and reserve capital $P_t$ for the overall belief process.
    In other words, if $P_{\tau_i} = p$ then $P_{\tau_i + t} = p_t$ for $t \in [0,\tau_{i+1}-\tau_i)$;
    at the next time of information acquisition, a new belief is drawn according to the experiment chosen: $P_{\tau_{i+1}} \sim F_{i+1}$.      

    The DM's information acquisition policy must be measurable with respect to the belief process and experiments must be Bayes plausible with respect to the current belief:
    \begin{align*}
    F_i \in \BayesPlausible(P_{{\tau_i}^-}) \text{ for all $i$, where } \BayesPlausible(p) := \left\{ F \in \Delta \Delta(\Theta) \middle| \int q dF(q) = p \right\}.
    \end{align*}
    Given prior $p$, adopt the convention that $P_{0^-}=p$ to accommodate for time zero information acquisition. 
    \hyref{Appendix}{sec:appendix:preliminaries} provides a rigorous construction of the belief process and admissible controls.


    \paragraph{Information costs} 
    Whenever information is acquired, the DM incurs a fixed cost and a variable cost. Given information acquisition policy $\{\tau_i,F_i\}$, at each $\tau_\id$ the DM pays the cost $C(F_i) + \kappa$, where $\kappa > 0$ is the fixed cost and the variable component $C:\Delta \Delta(\Theta) \rightarrow \overline{\reals}_+$ maps information structures into (non-negative, potentially infinite) costs.
      
    The variable component of information costs is \textbf{uniformly posterior separable} (UPS).
    Specifically, assume there exists a convex function $c: \Delta(\Theta) \rightarrow \overline{\reals}_+$, finite and continuously differentiable over the interior of $\Delta(\Theta)$, such that for any $F \in \Delta \Delta (\Theta)$:
    \begin{align*}
        C(F) = \int_{\Delta(\Theta)} \bigl( c(q) - c(p) \bigr) dF(q) \text{ where: } p = \int q dF(q).
    \end{align*}
    Common choices for $c$ include entropy (Shannon costs), negative variance, and the expected log-likelihood ratio.
    One possible interpretation of UPS costs is to see $c$ as a "measure of certainty" at a given belief; hence $C(F)$ corresponds to the \emph{expected increase in certainty} (reduction of uncertainty) induced by the chosen experiment relative to the current  belief. 
    See \citet{frankel2019quantifying} for a formalization of this interpretation, or \citet{caplin2022rationally,denti2022posterior,pomatto2023cost} for general references. 
    
    A natural foundation for such costs within the model's structure is to assume that primitive information costs come from dynamic evidence gathering with a time horizon which is \emph{negligible} relative to the time scale at which the state changes. 
    This can be formalized by taking the limit of a discrete-time model with two nested time scales, so that the state does not change \emph{while} the decision maker is acquiring information and the UPS cost corresponds to the induced reduced-form cost \citep[see][for how dynamic processes can induce UPS costs]{morris2019wald,denti2022experimental,bloedel2020cost}.
    Relatedly, UPS costs induce no intrinsic incentive to smooth information acquisition over time, because of linearity in the posterior distribution \citep[][establish that this property essentially identifies the class of UPS cost functions]{bloedel2020cost}. 
    Hence, this assumption isolates the changing state as the sole source of dynamics in the model: in the persistent state limit, information would be acquired at most once.

    \paragraph{Optimal information acquisition problem} 
    The DM chooses an information acquisition policy so as to maximize total discounted expected utility under exponential discounting at rate $r>0$.
    Hence, they solve the following \textbf{optimal information acquisition problem}:
    \begin{align}
    \label{eq:dynamic-problem} \tag{OIA}
        v(p) := \sup_{\substack{\{\tau_i, F_i\}_{i \geq 0}}} \Esp \Biggl[ \int_0^\infty e^{-rt} u(P_t) dt - \sum_{\id \geq 0} e^{-r {\tau_\id} } \bigl( C (F_\id) + \kappa \bigr) \Biggm| P_{0^-}=p \Biggr],
    \end{align}
    where $v(p)$ is the expected payoff from optimal information acquisition given initial belief $p$.

\section{Characterization of optimal policies}\label{sec:Characterization-Of-Solutions}

    The characterization of solutions in the optimal information acquisition problem relies on a familiar dynamic programming approach, which suggests that the value function solves the recursive (Bellman) equation:
    \begin{align*}
    \label{eq:recursive-equation} \tag{$\star$}
        v(p) = \sup_{\tau \geq 0} \hspace{.5em} \Biggl[ \int_0^\tau e^{-rt} u(p_t) dt + e^{-r \tau} \Biggl( \sup_{F \in \BayesPlausible(p_\tau)} \int_{\Delta(\Theta)} v dF - C(F) - \kappa \Biggr) \Biggr].
    \end{align*}
    \hyref{Section}{sec:FP} establishes this rigorously; \hyref{Section}{sec:optimal-policies} uses the implied decomposition between timing and content to derive the characterization of solutions.

    \subsection{Recursive equation and decomposition}\label{sec:FP}

        Standard dynamic programming logic delivers a formal derivation as well as intuition for \eqref{eq:recursive-equation}.
        Until the \emph{next} information acquisition time $\tau \geq 0$, beliefs drift deterministically and the DM accrues flow payoffs $u(p_t)$.
        At $\tau$, the DM incurs a cost $C(F) + \kappa$ and beliefs jump stochastically according to chosen experiment $F \in \BayesPlausible(p_\tau)$. 
        Given optimal behavior, the continuation value is the expected value function $\int v dF$, leading to the recursive equation \eqref{eq:recursive-equation}.

        From now on, denote by $\BellmanOperator$ the Bellman operator implicitly defined by the right-hand side of \eqref{eq:recursive-equation}, which maps any bounded measurable continuation payoff function to the induced value from one-shot information acquisition.
        By definition $v$ solves the recursive equation \eqref{eq:recursive-equation} if and only if it is a fixed point of $\BellmanOperator$, i.e. $\BellmanOperator v=v$. 
        $\BellmanOperator$ is not a contraction; to establish uniqueness, I decompose the Bellman equation into two operations and leverage the lattice structure of a suitably reduced domain of candidate value functions.

        \paragraph{Ex ante bounds on the value function} 
        Define the functions $\overline v$ and $\underline v$ as, respectively, the value from perfect costless observation of the state and from never getting any information about the true state, starting from an initial belief $p \in \Delta(\Theta)$ i.e.
        \begin{align*}
        \overline{v}(p) : = \int_0^\infty  e^{-rt} \Esp_{\theta \sim p_t} \biggl[ \max_a u(a,\theta) \biggr] dt,
        \hspace*{2em}
        \underline{v}(p) : = \int_{0}^\infty e^{-rt} u (p_t) dt.
        \end{align*}
        Let $\CandidateValue$ be the set of "candidate value functions": real-valued bounded measurable functions on $\Delta(\Theta)$ which are pointwise between $\underline{v}$ and $\overline{v}$; naturally $v \in \CandidateValue$.
        Foreshadowing results slightly, it is also convenient to define "net" bounds $\underline{w}(p):=\underline{v}(p)-c(p)$ and $\overline{w}(p):=\overline{v}(p)-c(p)$.
     
        \paragraph{Decomposition and uniqueness}
        The recursive equation can be decomposed into two parts: (i) the choice of an optimal information structure conditional on stopping, which reduces to an \emph{"as-if-static" information acquisition problem} where the continuation values $v$ itself plays the role of the indirect utility function; (ii) the choice of the optimal timing of information acquisition, which reduces to an \emph{"as-if-one-off" deterministic optimal stopping problem} where the stopping payoff is given by value from instaneous information acquisition net of the fixed cost.
        The recursive equation requires that ("static") value from information incorporate future value from optimal timing and ("one-off") stopping payoffs derive from future information.

        To prove that $v$ is the unique fixed point of $\BellmanOperator$, I rely on properties of the functional operators corresponding to the value in the static information problem and the optimal stopping problem, given arbitrary continuation values. 
        These operators are monotone and order-convex operators over the lattice of real-valued bounded measurable functions over $\Delta(\Theta)$, which is a Riesz space, of which $\CandidateValue$ is an order-interval.
        This enables the use of a Tarski-style fixed point theorem from \citet{marinacci2019unique}, where convexity and an upper-perimeter conditions deliver uniqueness.
        Details of the proof are in \hyref{Appendix}{sec:appendix:FP-existence-uniqueness}.

        \begin{theorem}
        \label{thm:recursive-characterization}
           The value function $v$ in the optimal information acquisition problem \eqref{eq:dynamic-problem} is the unique solution to the recursive equation \eqref{eq:recursive-equation} in $\CandidateValue$, and is continuous.
        \end{theorem}

        The value function is convex if $u$ is convex.
        Iterations of the fixed point operator $\BellmanOperator$ initalized at the lower bound $\underline{v}$ provide some economic intuition. 
        Indeed, they correspond to the value in a \emph{constrained} problem, where the DM is only allowed a finite number $n$ of times of information acquisition, and converge to $v$ as $n$ goes to infinity.

    \subsection{Optimal policies}\label{sec:optimal-policies}

        \paragraph{Intuitive derivation}
        Define the payoff from immediately and optimally acquiring information under continuation value $v$, gross of fixed cost:
        \begin{align*}
         \InfoVal v(p) := \sup_{F \in \BayesPlausible(p)} \int v dF - C(F).
        \end{align*}
        At any $p$ which triggers updating of information the value must equal the payoff from stopping, so $v(p) = \InfoVal v(p) - \kappa$.
        By definition $\InfoVal v \geq v$ from static optimality; simultaneously dynamic programming entails that $v \geq \InfoVal v - \kappa$ (immediate information acquisition is always feasible).
        Hence, in the dynamic context, the difference $\InfoVal v(p)-v(p) \in [0,\kappa]$ captures the \emph{residual} gross value of information.
        The Bellman operator $\BellmanOperator v$ gives the value in the optimal stopping problem with terminal payoff $\InfoVal v$, so $\BellmanOperator v- v$ is the interim value of one-shot information acquisition. 
        The value function $v$ is the unique candidate function with zero such interim value: $v$ is unimprovable from one shot information acquisition ($v \geq \BellmanOperator v$) and is meanwhile an attainable continuation value ($\BellmanOperator v \geq v$).
        
        Optimal experiments can also be expressed in terms of the residual gross value $\InfoVal v(p)-v(p)$.   
        This follows from well-known results on \emph{static} information acquisition with UPS costs \citep[see e.g.][]{caplin2022rationally,gentzkow2014costly,dworczak2024persuasion}.
        For a given belief $p$, the optimal experiment is derived by identifying the chord between two points of the graph of $v - c$ that attains the highest value at $p$.  
        This follows from the problem's linear structure (in $F$, when fixing $p$), and implies the following "concavification" characterization:
            \[
            \InfoVal v(p) - c(p) = \sup_{F \in \BayesPlausible(p)} \int [v - c]dF = \Cav[v-c](p),
            \]
        where $\Cav$ denotes the upper concave envelope operator.
        Optimal binary experiments can be identified with their support and for any $p$ the optimal experiment is supported on the closest points such that $\Cav[v-c]=v-c$.
        Hence if the support $\{q_0,q_1\}$ is optimal at $p$, it remains optimal for any $q \in (q_0, q_1)$, thus partitioning the belief space into non-overlapping "experiment intervals", the endpoints of which outline optimal experiments.
        
        As a result, the residual value $\InfoVal v - v$ sufficiently characterizes optimal information acquisition and is itself pinned down by \emph{net} value function $w:=v-c$.
        To further highlight the logic, observe that the Bellman equation can be rewritten on $w$, with explicit continuation value:
        \begin{equation}
            \label{eq:net-recursive-equation}
            \tag{$\hat \star$}
            w(p) = \sup_{\tau \geq 0} \int_0^\tau e^{-rt} f(p_t) + e^{-r\tau} \Bigl( \Cav[w](p_\tau) - \kappa \Bigr),
        \end{equation}
        where $f(p):= u(p) - rc(p) + \lambda (\pi-p) c'(p)$ denotes the "virtual flow payoff".

        \paragraph{Formal result}
        Putting everything together, the general result gives a geometric characterization of solutions, which implicitly defines a simple class of policies parameterized by two nested collections of intervals.
        To make statements more concise, let $\CavDist w:= \Cav w-w$ the residual value of information operator (given $v$ the unique fixed point of $\Phi$ and $w:=v-c$) and define:
        \begin{align*}
            \grossvalregionopt := \biggl\{ p \in \Delta(\Theta) \biggm| \CavDist w(p) > 0\biggr\},
            \quad
            \inforegionopt := \biggl\{ p \in \Delta(\Theta) \biggm| \CavDist w(p) = \kappa \biggr\}.
        \end{align*}
        By definition $\grossvalregionopt$ is a countable collection of disjoint open intervals and $\inforegionopt \subset \grossvalregionopt$.

        \begin{theorem}
        \label{thm:optimal-policies}
            The following information acquisition policy is optimal:
            \begin{enumerate}
                \item acquire information whenever the residual value of information equals the fixed cost, which is described by the waiting time:
                        \begin{align*}
                            \tau^*(p) := \inf \bigl\{ t \geq 0 \bigm| p \in \inforegionopt \bigr\},
                        \end{align*}
                \item if information is acquired at $p$ choose the binary experiment $F^*_p$ supported over the two closest points to $p$ not in $\grossvalregionopt$ (i.e. at which there is no residual value of information: $\CavDist w = 0$).
            \end{enumerate}
        \end{theorem}    

        The statement of \hyref{Theorem}{thm:optimal-policies} nests two points. 
        The first is the characterization of optimal policy purely in terms of the \emph{residual} value of information $\InfoVal v - v$.
        The second is the explicit form for this residual value in terms of the net value function $w=v-c$.
        Together, they reduce the description of optimal information acquisition to the choice of \emph{two nested regions}: the "collection of experiment intervals" $\grossvalregionopt$ and the "information acquisition region" $\inforegionopt$.

        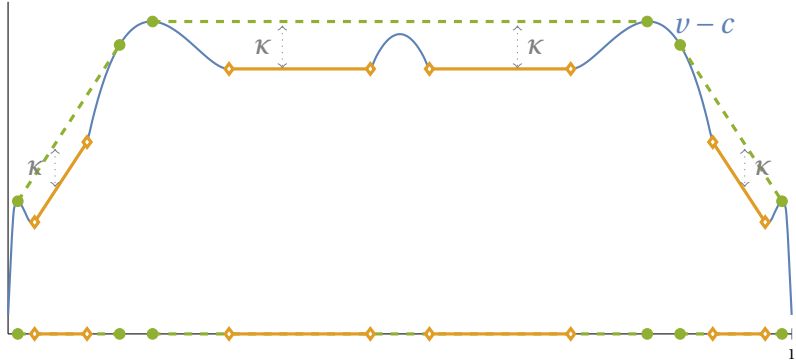
\begin{figure}[h!]
            \centering
            \pgfplotsset{
            table/search path={Tikz/Tikz_Portfolio_Example2},
            }
            \begin{tikzpicture}
                \begin{axis}[
                  xmin=0, xmax=1,
                  ymin=1.79, ymax=1.86,
                  clip = false,
                  height=6cm,
                  width=12cm,
                  xtick distance = 1,
                  ytick distance = 1,
                  ticklabel style = {font=\tiny},
                  ylabel={},
                  axis lines=middle,
                  axis line style={-},
                  legend style={at={(.5,.7)},anchor=west,font=\footnotesize,draw=none}
                  ]

                  \addplot+ [mathematicablue, thick, no marks, smooth] table {v_net.table} node[right,pos=.82] {$v-c$};
            
                  \addplot [mathematicagreen,mark=*,only marks,mark options={color=mathematicagreen,fill=mathematicagreen}] table {targets_vnet.table};
                  \addplot [mathematicagreen,mark=*,only marks,mark options={color=mathematicagreen,fill=mathematicagreen}] table {targets_vnet_axis.table};
                  \addplot [mathematicagreen,no marks,very thick,dashed] table {targets1_vnet_axis.table};
                  \addplot [mathematicagreen,no marks,very thick,dashed] table {targets2_vnet_axis.table};
                  \addplot [mathematicagreen,no marks,very thick,dashed] table {targets3_vnet_axis.table};
                  \addplot [mathematicagreen,no marks,very thick,dashed] table {targets1_vnet.table};
                  \addplot [mathematicagreen,no marks,very thick,dashed] table {targets2_vnet.table};
                  \addplot [mathematicagreen,no marks,very thick,dashed] table {targets3_vnet.table};

                  \addplot [mathematicaorange,mark=diamond*,very thick,mark options={color=mathematicaorange,fill=white}] table {thresholds1_vnet.table};
                  \addplot [mathematicaorange,mark=diamond*,very thick,mark options={color=mathematicaorange,fill=white}] table {thresholds2_vnet.table};
                  \addplot [mathematicaorange,mark=diamond*,very thick,mark options={color=mathematicaorange,fill=white}] table {thresholds3_vnet.table};
                  \addplot [mathematicaorange,mark=diamond*,very thick,mark options={color=mathematicaorange,fill=white}] table {thresholds4_vnet.table};
                  \addplot [mathematicaorange,mark=diamond*,smooth,very thick,mark options={color=mathematicaorange,fill=white}] table {thresholds1_vnet_axis.table};
                  \addplot [mathematicaorange,mark=diamond*,smooth,very thick,mark options={color=mathematicaorange,fill=white}] table {thresholds2_vnet_axis.table};
                  \addplot [mathematicaorange,mark=diamond*,smooth,very thick,mark options={color=mathematicaorange,fill=white}] table {thresholds3_vnet_axis.table};
                  \addplot [mathematicaorange,mark=diamond*,smooth,very thick,mark options={color=mathematicaorange,fill=white}] table {thresholds4_vnet_axis.table};

                  \addplot [<->, gray, dotted] coordinates {(.35,1.855)  (.35,1.846)} node[left,pos=.5] {$\kappa$};
                  \addplot [<->, gray, dotted] coordinates {(.65,1.855)  (.65,1.846)} node[right,pos=.5] {$\kappa$};
                  \addplot [<->, gray, dotted] coordinates {(.06,1.829)  (.06,1.821)} node[left,pos=.5] {$\kappa$};
                  \addplot [<->, gray, dotted] coordinates {(.94,1.829)  (.94,1.821)} node[right,pos=.5] {$\kappa$};

                \end{axis}
            \end{tikzpicture}
            \caption{Geometrically solving for the optimal policy given the net value function $v-c$}
            {\footnotesize The green intervals represent $\grossvalregionopt$ (gross static value for information); the orange region represent $\inforegionopt$ (information acquisition region); the green dots are the local "target" beliefs when information is acquired.}
            \label{fig:geometric-intuition}
        \end{figure}  
        
        \paragraph{Geometric representation}
        \hyref{Theorem}{thm:optimal-policies} induces a simple parameterization and a useful geometric visualization for optimal policies using the graph of $v-c$, which is illustrated in \hyref{Figure}{fig:geometric-intuition}.
        First, draw the \emph{concave envelope} of $v-c$ and look for the region where it is strictly above $v-c$: this gives $\grossvalregionopt$.
        Each interval in $\grossvalregion$ defines a region where information \emph{may} be acquired and the corresponding optimal binary experiment supported its endpoints.
        Within each interval, look for the subset of beliefs at which $v-c$ coincides with its own concave envelope shifted down by $\kappa$; note that this need not be an interval itself as it may have "holes".
             
        The proof of \hyref{Theorem}{thm:optimal-policies} formalizes the logic previously described.
        It follows from combining three results that can be found in \hyref{Appendix}{sec:appendix:optimal-policies}.
        First, \hyref{Proposition}{prop:concavification} characterizes optimal experiments for arbitrary continuation values, which yields the sufficiency of binary experiments and a the interval decomposition, as well a general version of the concave envelope characterization.
        Second, \hyref{Proposition}{prop:optimal-stopping-facts} characterizes optimal timing of one shot information acquisition, taking as a given the value from optimal information acquisition.
        Third, putting it together is justified by a general verification result characterizing all optimal policies in terms of solutions in the Bellman equation (\hyref{Proposition}{prop:verification-theorem-general-form}).

        \paragraph{Remark} 
        In the remainder of the paper, optimal information acquisition refers to the optimal policy in \hyref{Theorem}{thm:optimal-policies}. 
        This policy is always well-defined, but cannot in general be guaranteed to be unique, although non-uniqueness can reasonably be expected to be a knife-edge case.
        In case of multiplicity, it corresponds to the selection of the earliest optimal stopping time and the least informative optimal experiment.

\section{Dynamics of information acquisition}\label{sec:Dynamics}

    The decomposition in \hyref{Theorem}{thm:optimal-policies} enables the precise description of dynamics.
    In the long run, learning must either stop in finite time or settle into a simple cyclical pattern (\hyref{Theorem}{thm:long-run-dynamics}).
    \hyref{Proposition}{prop:when-does-learning-stop} characterizes conditions under which learning stops and path-dependent "learning traps" exist.
    The results deliver a simplified problem over explicit stationary payoffs (\hyref{Section}{sec:stationary-payoffs-and-auxiliary-problem}) and a natural way to compare informativeness across environments (\hyref{Section}{sec:Frequency-Quality-CS}).
  
    \subsection{Convergence to cyclical information acquisition}\label{sec:convergence-to-cyclical-information-acquisition}

        \paragraph{Belief cycles} 
        Under the optimal strategy, if at any point an experiment is chosen which leads to possible posterior beliefs $\lowtarget,\hightarget$ on opposite sides of the long run average $\pi$ ($\lowtarget < \pi < \hightarget$), then all future information acquisition leads to the same two posteriors. 
        By \hyref{Theorem}{thm:optimal-policies}, the same support must remain optimal for the next experiment (see \hyref{Figure}{fig:geometric-intuition}).
        This implies that the time between updates is the waiting time between either $\lowtarget,\hightarget$ and the closest belief towards $\pi$ which lies in the information acquisition region $\inforegionopt$.
        In other words, if information acquisition is supported on beliefs which suggest that a different state is more likely than average, beliefs must enter simple cyclical dynamics, where a fixed time between updates lead to restoring one of two possible levels of relative confidence that one state is more likely than average.
        \hyref{Definition}{def:belief-cycles} below formalizes this notion of "belief cycles";  \hyref{Figure}{fig:belief-cycles} illustrates cyclical dynamics.

        \begin{figure}[h!]
            \centering
            \begin{tikzpicture}[xscale=8,yscale=6,>=stealth]
                \draw [|-|] (0,0) -- (1,0);    

                \node[circle, draw=mathematicagreen, fill=mathematicagreen, inner sep=0pt, minimum size=5pt, label={[label distance=-.5em]170:{\textcolor{mathematicagreen}{$\lowtarget$}}}] (q0) at (.15,0) {};
                \node[circle, draw=mathematicagreen, fill=mathematicagreen, inner sep=0pt, minimum size=5pt, , label={[label distance=-.4em]10:{\textcolor{mathematicagreen}{$\hightarget$}}}] (q1) at (.85,0) {};

                \node[diamond, draw=mathematicaorange, fill=mathematicaorange, inner sep=0pt, minimum size=5pt, label=above:{\textcolor{mathematicaorange}{$\lowthreshold$}}] (pL) at (.3,0) {};
                \node[diamond, draw=mathematicaorange, fill=mathematicaorange, inner sep=0pt, minimum size=5pt, label=below:{\textcolor{mathematicaorange}{$\highthreshold$}}] (pH) at (.7,0) {};


                \node[circle, draw=mathematicablue, fill=mathematicablue, inner sep=0pt, minimum size=3pt, label=below:{\textcolor{mathematicablue}{$\pi$}}] (pi) at (.5,0) {};
                
                \path [thick, mathematicaorange, dashed, ->] (pL) edge[bend left =90]  node{} (q0);
                \path [thick, mathematicaorange, dashed, ->] (pL) edge[bend left =-90]  node{} (q1);

                \path [thick, mathematicaorange, dashed, ->] (pH) edge[bend right =90]  node{} (q0);
                \path [thick, mathematicaorange, dashed, ->] (pH) edge[bend right =-90]  node{} (q1);

                \path[left color=mathematicagreen,right color=mathematicaorange] (.15,.003) rectangle (.3,-.003);
                \path[right color=mathematicagreen,left color=mathematicaorange] (.85,.003) rectangle (.7,-.003);

                \begin{scope}[every node/.style={draw,fill,single arrow,mathematicablue,
                    single arrow tip angle=50,
                    single arrow head extend=2pt,
                    single arrow head indent=1pt,
                    inner sep=0pt}] 
                  \node[shape border rotate=0] at (.225,0) {};
                  \node[shape border rotate=180] at (.775,0) {};
                \end{scope}
            \end{tikzpicture}
            \caption{Representation of a belief cycle}
            \label{fig:belief-cycles}
        \end{figure}
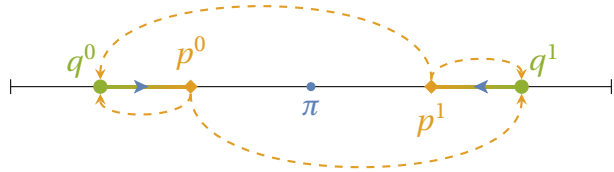
        
        \begin{definition}\label{def:belief-cycles}
              A \emph{belief cycle} $\beliefcycle$ is a tuple $\beliefcycle=\biggl( (\lowtarget,\hightarget) , (\lowthreshold,\highthreshold) , (\lowtime,\hightime) \biggr)$ composed of target beliefs $(\lowtarget,\hightarget)$, threshold beliefs $(\lowthreshold,\highthreshold)$, and waiting times $(\lowtime,\hightime)$ such that:
                    \begin{gather*}
                    0 \leq \lowtarget \leq \lowthreshold \leq \pi \leq \highthreshold \leq \hightarget \leq 1
                    \text{ and } \tau_i = \frac{1}{\lambda} \log \Biggl( \frac{\pi-q^i}{\pi-p^i} \Biggr) \hspace{.5em} \text{ for } i=0,1
                    \end{gather*}      
                It is called \emph{non-degenerate} if $\lowtarget < \lowthreshold < \pi < \highthreshold < \hightarget$ or, equivalently, $\lowtarget < \pi < \hightarget$ and $\lowtime,\hightime>0$
        \end{definition}

        A belief cycle parameterizes the law of motion of beliefs within the induced domain $[\lowtarget,\lowthreshold] \cup [\highthreshold,\hightarget]$: information is acquired when beliefs reach $\lowthreshold$ or $\highthreshold$, the resulting update triggers a jump to $\lowtarget$ or $\hightarget$; the DM waits $\lowtime$ or $\hightime$ respectively until the next update (see \hyref{Figure}{fig:belief-cycles}).
        This description of the belief cycle therefore has some redundancy but it conveniently encodes the dependence on parameters $(\lambda,\pi)$ for comparison across environments (see \hyref{Section}{sec:CS-volatility}).
          
        \paragraph{Convergence}
        Information acquisition must eventually either reach a belief cycle or stop altogether. 
        \hyref{Figure}{fig:geometric-intuition} and the examples of belief dynamics in \hyref{Figure}{fig:optimal-belief-dynamics-examples} below provide intuition into the underlying logic: if both posteriors from a given experiment are on the same side of $\pi$, then if beliefs jump \emph{towards} $\pi$ they will only drift closer to $\pi$ until they reach an experiment which triggers cyclical dynamics (if there is one).
        \begin{theorem}\label{thm:long-run-dynamics}
          Let $\{P_t\}$ the belief process deriving from optimal information acquisition. There exists an almost surely finite time $T \geq 0$ after which either:
          \begin{enumerate}[label=\textbf{\emph{(\Alph*)}}]
              \item \textbf{Learning stops:} no information is acquired, or
              \item \textbf{Cyclical updates:} $P_t$ follows dynamics described by a non-degenerate belief cycle.
          \end{enumerate}
          If learning stops, beliefs converge to their long-run average: $P_t \xrightarrow[t \rightarrow \infty]{a.s.} \pi$.
        \end{theorem}

        \begin{figure}[h!]
          \centering
            \begin{subfigure}{.4\textwidth}
                \includegraphics[scale=.4]{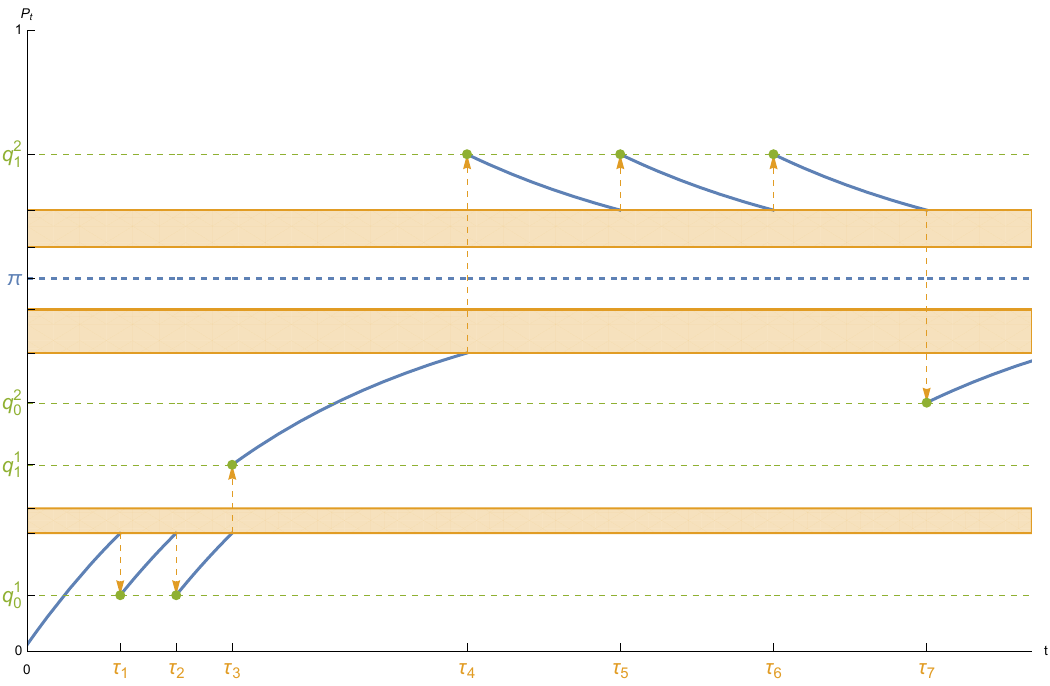}
                \caption{Convergence to cyclical dynamics}
            \end{subfigure}
            \hspace{2em}
            \begin{subfigure}{.4\textwidth}
                \includegraphics[scale=.4]{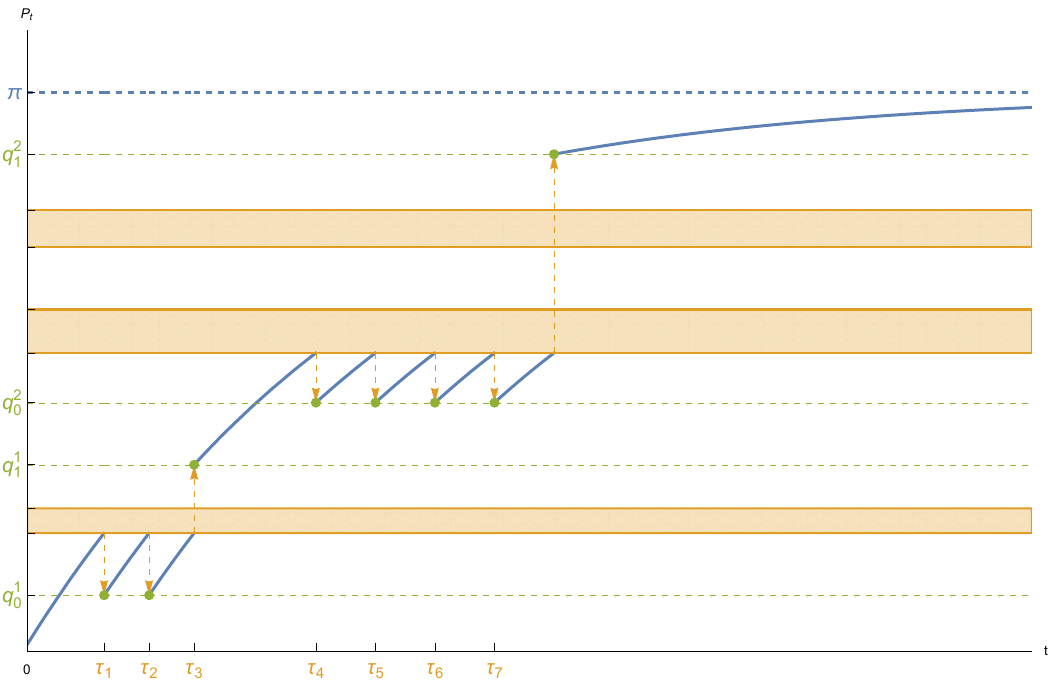}
                \caption{Learning eventually stops}
            \end{subfigure}    
            \caption{Examples of realized belief dynamics}
            \label{fig:optimal-belief-dynamics-examples}
        \end{figure}

        The proof would be fairly direct if $\grossvalregionopt$ could be guaranteed to have a \emph{finite} numbers of disjoint intervals, but there could be infinitely many.
        The general proof (in \hyref{Appendix}{sec:appendix:dynamics-long-run-behavior}) is an application of Kolmogorov's 0-1 Law.
        The theorem implies that the DM eventually settles on \emph{periodic updates} with \emph{fixed content}.
        This rules out more complex yet reasonable behavior -- for instance, frequent "small" updates to check whether a more substantial reassessment should occur.  
        This reduction of the problem to the choice of \emph{content} and \emph{frequency} extends the intuition that static information acquisition optimally concentrates on two possible beliefs, each suggesting one state being relatively more likely than the prior.
        In the dynamic problem, it is optimal to \emph{eventually} concentrate all information acquisition on two possible outcomes, which must now suggest either state being more likely than \emph{the long run average}.
        The substantial difference is that the value of information stems from the periodic repetition of that experiment: the benefit from information acquisition is the sum of short-run improvements in decisions. 
        Hence, it is necessary to specify not just content but also frequency, and choices over these objects are entangled: periodicity determines both occurrence of costs and time horizon (hence short-run value) from each update.

        \paragraph{Short and long run}
        The convergence result in \hyref{Theorem}{thm:long-run-dynamics} clarifies incentives for short-run (non-cyclical) information acquisition: it entails a gamble on \emph{when} the process enters the stationary cycles. 
        Suppose the DM has solved for the optimal stationary dynamics: once in $[\lowtarget,\lowthreshold]\cup[\highthreshold,\hightarget]$, it is optimal to remain in the corresponding cycle. 
        Fix a starting belief $p<\lowtarget$ and allow at most one experiment region outside the stationary domain. 
        The DM chooses a stopping belief in $(p,\lowtarget)$ and a local two-point experiment with both posteriors to the left of $\lowtarget$. 
        A jump to the left triggers a temporary local loop; a jump to the right leads to drift toward the stationary regime with continuation value $v(\lowtarget)$. 
        Thus short-run acquisition trades off lingering at higher-certainty beliefs against jumping sooner to the steady state (lower certainty).
        This suggests a constructive decomposition: first solve the stationary policy; then solve the nearest short-run acquisition region taking the inward continuation value as given; and iterate outward.\footnote{This induction is well-defined under additional regularity, e.g. if $u$ and $c$ are piecewise analytic, to guarantee that there are only finitely many intervals where acquisition is profitable.}
        The logic also explains departures from static intuition: information can be optimal even without an immediate action switch if it changes the \emph{timing} of future switches. 
        In the long-run periodic regime this phenomenon vanishes: no information is acquired without action switches—otherwise the auxiliary problem would entail costs without benefits, echoing the static case.

        \paragraph{Ergodic distribution of beliefs}
          The result on long run belief dynamics also enables the characterization of the \emph{ergodic distribution of beliefs}; the proof is standard and omitted.
          \begin{proposition}\label{prop:ergodic-distribution}
            Let $\mu$ the ergodic distribution of beliefs under optimal information acquisition, identified with its density. Assume information acquisition does not stop under the optimal policy and denote $[\lowtarget,\lowthreshold] \cup [\highthreshold,\hightarget]$ the support of the long run belief cycle. Then $\mu$ is \emph{piecewise uniform} with $\mu(p) = \frac{1}{2} \frac{1}{\lowthreshold-\lowtarget}$ if $p \in [\lowtarget,\lowthreshold]$ and $\mu(p)=\frac{1}{2} \frac{1}{\hightarget-\highthreshold}$ if $p \in [\highthreshold,\hightarget]$.
          \end{proposition}           
          The ergodic distribution can be interpreted as the eventual distribution of beliefs within a population of identical decision-makers; all objects of interest (spread of beliefs, average time to the next update,...) express in terms of the thresholds.
          This, in turn, enables using the model to study the effect of optimal information acquisition and the model’s primitives on population parameters, e.g. the spread and balance of beliefs within the population, and investigate broader questions such as whether higher information costs lead to more or less disagreement.

    \subsection{Learning traps}\label{sec:learning-traps}

        When do optimal dynamics induce learning to stop in finite time? 
        The following proposition gives simple conditions depending on whether or not there is gross value for information at $\pi$. 
    
        \begin{proposition}\label{prop:when-does-learning-stop}
            Under optimal dynamics:
            \begin{enumerate}[label=(\roman*)]
                \item Learning stops in finite time only if $\pi$ is not in the information acquisition region ($\pi \notin \inforegionopt$).
                \item If there no gross value for information at $\pi$ (i.e. $\pi \notin \grossvalregionopt$), information is acquired at most a finite number of times.
                \item If there is gross but not net value for information at $\pi$ ($\pi \in \grossvalregionopt \setminus \inforegionopt$), then either (a) information acquisition is acquired at most a finite number of times for all priors, or (b) there exists some open interval $(\underline{p},\overline{p})$ such that no information is ever acquired for any prior $p_0 \in (\underline{p},\overline{p})$ but any prior \emph{not} in $(\underline{p},\overline{p})$ leads to a belief cycle in the long run.
            \end{enumerate}
            Furthermore case (iii.b) occurs if and only if there exists some beliefs in the information acquisition region to the left and to the right of $\pi$ which both lead to the same conditionally optimal target beliefs as at $\pi$. Refer to the (possibly empty) interval $(\underline{p},\overline{p})$ as the "trap region".
        \end{proposition}

        Whether learning stops depends on whether the belief process drifts into points in the information acquisition region, coming from any side of $\pi$. 
        \hyref{Proposition}{prop:when-does-learning-stop} follows this logic to delineate necessary and sufficient conditions. 
        Sufficient conditions can often be obtained without solving the full problem, by examining the best attainable payoffs from \emph{some} cycle initiated at or around $\pi$.
        For instance: if some cycle is profitable, then it must be that $\pi \in \inforegionopt$ and all priors eventually lead to periodic information acquisition; if the same is true in gross but not net value, there is a hole and path dependence arises; etc.

        \begin{figure}[h!]
            \centering
            \begin{tikzpicture}[xscale=9,yscale=9.5,>=stealth]
                \draw [-] (0,0) -- (1,0);
                \draw [dashed] (0,0) -- (-.1,0);
                \draw [dashed] (0,0) -- (1.1,0);

                \node[circle, draw=mathematicagreen, fill=mathematicagreen, inner sep=0pt, minimum size=5pt, label=above:{\color{mathematicagreen} $\lowtarget$}] (pL2) at (.2,0) {};
                \node[circle, draw=mathematicagreen, fill=mathematicagreen, inner sep=0pt, minimum size=5pt, label=above:{\color{mathematicagreen} $\hightarget$}] (pH2) at (.95,0) {};
                \draw[thick, mathematicagreen] (pL2) -- (pH2);

                \node[diamond, draw=mathematicaorange, fill=mathematicaorange, inner sep=0pt, minimum size=4pt, label=above:{\color{mathematicaorange} $\lowthreshold$}] (pl21) at (.4,0) {};
                \node[diamond, draw=mathematicaorange, fill=mathematicaorange, inner sep=0pt, minimum size=4pt, label=above:{\color{mathematicaorange}$\highthreshold$}] (ph22) at (.75,0) {};
                \draw[thick,mathematicaorange] (pl21) -- (ph22);

                \node[cross out, draw=mathematicablue, fill=mathematicablue, very thick, inner sep=0pt, minimum size=2pt, label=below:{\textcolor{mathematicablue}{$\pi$}}] (pi) at (.6,0) {}; 
            \end{tikzpicture}
            \\
            \begin{tikzpicture}[xscale=9,yscale=9.5,>=stealth]
                \draw [-] (0,0) -- (1,0);
                \draw [dashed] (0,0) -- (-.1,0);
                \draw [dashed] (0,0) -- (1.1,0);

                \node[circle, draw=mathematicagreen, fill=mathematicagreen, inner sep=0pt, minimum size=5pt, label=above:{\color{mathematicagreen}$\lowtarget$}] (pL2) at (.05,0) {};
                \node[circle, draw=mathematicagreen, fill=mathematicagreen, inner sep=0pt, minimum size=5pt, label=above:{\color{mathematicagreen}$\hightarget$}] (pH2) at (.9,0) {};
                \draw[thick, mathematicagreen] (pL2) -- (pH2);


                \node[diamond, draw=mathematicaorange, fill=mathematicaorange, inner sep=0pt, minimum size=4pt, label=above:{\color{mathematicaorange}$\lowthreshold$}] (pl21) at (.4,0) {};
                \node[diamond, draw=mathematicaorange, fill=mathematicaorange, inner sep=0pt, minimum size=4pt, label=above:{\color{mathematicaorange}$\highthreshold$}] (ph22) at (.85,0) {};
                \draw[thick,mathematicaorange] (pl21) -- (ph22);

                \node[cross out, draw=mathematicablue, fill=mathematicablue, very thick, inner sep=0pt, minimum size=2pt, label=below:{\textcolor{mathematicablue}{$\pi$}}] (pi) at (.3,0) {};

            \end{tikzpicture}
            \\
            \begin{tikzpicture}[xscale=9,yscale=9.5,>=stealth]
                \draw [-] (0,0) -- (1,0);
                \draw [dashed] (0,0) -- (-.1,0);
                \draw [dashed] (0,0) -- (1.1,0);

                \node[circle, draw=mathematicagreen, fill=mathematicagreen, inner sep=0pt, minimum size=5pt, label=above:{\color{mathematicagreen}$\lowtarget$}] (pL2) at (.1,0) {};
                \node[circle, draw=mathematicagreen, fill=mathematicagreen, inner sep=0pt, minimum size=5pt, label=above:{\color{mathematicagreen}$\hightarget$}] (pH2) at (.9,0) {};
                \draw[thick, mathematicagreen] (pL2) -- (pH2);

                \node[diamond, draw=mathematicaorange, fill=mathematicaorange, inner sep=0pt, minimum size=4pt, label=above:{\color{mathematicaorange}$\lowthreshold$}] (pl21) at (.2,0) {};
                \node[diamond, draw=mathematicaorange, fill=mathematicaorange, inner sep=0pt, minimum size=4pt, label=below:{\color{myred}$\underline{p}$}] (ph21) at (.35,0) {};
                \draw[thick,mathematicaorange] (pl21) -- (ph21);

                \node[diamond, draw=mathematicaorange, fill=mathematicaorange, inner sep=0pt, minimum size=4pt, label=below:{\color{myred} $\overline{p}$}] (pl22) at (.65,0) {};
                \node[diamond, draw=mathematicaorange, fill=mathematicaorange, inner sep=0pt, minimum size=4pt, label=above:{\color{mathematicaorange}$\highthreshold$}] (ph22) at (.8,0) {};
                \draw[thick,mathematicaorange] (pl22) -- (ph22);


                \node[cross out, draw=mathematicablue, fill=mathematicablue, very thick, inner sep=0pt, minimum size=2pt, label=below:{\textcolor{mathematicablue}{$\pi$}}] (pi) at (.5,0) {};

            \end{tikzpicture}
            \caption{Possible cases from \hyref{Proposition}{prop:when-does-learning-stop}}
            {\footnotesize The green interval represents $\grossvalregionopt$ and the orange region is $\inforegionopt$ (locally around $\pi$). 
            \\
            Case $1$ leads to a belief cycle for all priors. Case $2$ leads to learning stopping for all priors. Case $3$ leads to a cycle for all priors outside of $(\underline{p},\overline{p})$ and no information acquisition otherwise.
            }
            \label{fig:long-run-belief-pairs}
        \end{figure}   

        The most interesting case is the last one, (iii.b): path dependence arises, whereas in all other cases, all beliefs lead to the same long run outcome. 
        When there is a "hole" around $\pi$, path dependence takes a stark form: no information is ever acquired if the DM started with uncertain enough beliefs; even though there is gross value of information at $\pi$, it is not sufficiently high to warrant paying the initial cost of cyclical information acquisition.  
        Even though this results from optimality, such path dependence might have broader welfare consequences, e.g. in the presence of externalities from information acquisition or if information inequality is undesirable.

    \subsection{The simplified stationary problem}\label{sec:stationary-payoffs-and-auxiliary-problem}

        \paragraph{Stationary payoffs} 
        Consider an arbitrary belief cycle $\beliefcycle=\bigl( (\lowtarget,\hightarget),(\lowthreshold,\highthreshold),(\lowtime,\hightime) \bigr)$ and the associated dynamics.
        Let $w^0,w^1$ the net values -- as in the modified recursive equation \eqref{eq:net-recursive-equation} -- for a DM in the cyclical dynamics starting from $\lowtarget,\hightarget$ respectively.
        By definition they must verify:
        \begin{equation}
        \label{eq:cyclical-payoffs}\tag{CP}
            \begin{aligned}
                w^\theta & = \int_0^{\tau^\theta} e^{-rt} f(q^\theta_t) dt + e^{-r \tau^\theta} \Biggl( \frac{\hightarget-p^\theta}{\hightarget-\lowtarget} w^0 + \frac{p^\theta-\lowtarget}{\hightarget-\lowtarget} w^1 - \kappa \Biggr) 
            \end{aligned}
        \end{equation}
        for $\theta=0,1$.
        This system has a unique solution $\bigl(w^1(\beliefcycle),w^0(\beliefcycle)\bigr)$ for any belief cycle $\beliefcycle$ (including degenerate belief cycles with the continuous extension $w^0(\beliefcycle)=w^1(\beliefcycle)=f(\pi)/r$ when $q^0=q^1=\pi$). 
        A further simplification comes from focusing on the induced value $w^\pi$ from jumping into the cycle $\beliefcycle$ from $\pi$:
        \begin{align*}
         w^\pi:= \frac{\hightarget-\pi}{\hightarget-\lowtarget} w^0 + \frac{\pi-\lowtarget}{\hightarget-\lowtarget} w^1.
         \end{align*}
         Note that the recursive equation above rewrites as:
         \begin{align*}
         w^\theta = \int_0^{\tau^\theta} e^{-rt} f(q^\theta_t)dt + e^{-r \tau^\theta} \Biggl( e^{-\lambda \tau^{\theta}} w^\theta + \bigl( 1-e^{-\lambda \tau^{\theta}} \bigr) w^\pi - \kappa \Biggr).
         \end{align*}
        This has convenient interpretation: continuation values along the cycle are weighted averages of the long-term average value $w^\pi$ and the one-sided conditional values $w^0$ and $w^1$, where weights depend on the time until the cycle is reset. 
        Plugging into the definition of $w^\pi$ for $\theta=0,1$ yields a single-variable equation that gives the following explicit expression:
        \begin{align*}
            w^\pi = \alpha \times \frac{\displaystyle \int_0^{\tau^1} e^{-rt} f(q^1_t)dt - e^{-r \tau^1} \kappa}{1-e^{-r \tau^1}} + (1-\alpha) \times \frac{\displaystyle \int_0^{\tau^0} e^{-rt} f(q^0_t)dt - e^{-r \tau^0} \kappa}{1-e^{-r\tau^0}},
        \end{align*}
        where:
        \begin{align*}
            \alpha:= \frac{\frac{\pi-q^0}{q^1-q^0} \frac{1-e^{-r \tau^1}}{1-e^{-(r+\lambda)\tau^1}}}{\frac{\pi-q^0}{q^1-q^0}  \frac{1-e^{-r\tau^1}}{1- e^{-(r+\lambda)\tau^1}} + \frac{q^1-\pi}{q^1-q^0} \frac{1-e^{-r\tau^0}}{1- e^{-(r+\lambda)\tau^0}}} \in [0,1].
        \end{align*}
        In other words, $w^\pi$ is a weighted average of the one-shot payoffs along each side of the cycle.
        From now on, denote $w^\pi(\beliefcycle)$ the long-run average value for any cycle $\beliefcycle$.
        
        \paragraph{Stationary problem} 
        Eventual stationary behavior in the \emph{dynamic} problem must coincide with the optimal belief cycle from optimizing over $w^\pi(\beliefcycle)$; this provides a tractable non-recursive approach to study properties of optimal long-run behavior.
        \begin{proposition}
            \label{prop:directly-stationary-problem}
              Any belief cycle $\beliefcycle$ which maximizes $w^\pi(\beliefcycle)$ is long-run optimal in the dynamic information acquisition problem, and conversely. 
              Furthermore:
              \begin{align*}
              w(\pi) = \max \Biggl\{  f(\pi), \hspace{.5em} \max_{\beliefcycle} w^\pi(\beliefcycle) - \kappa \Biggr\}
              \end{align*} 
              and belief $p$ is contained in the trap region $(\underline{p},\overline{p})$ if and only if: 
              \begin{align*}
              \max_{\beliefcycle} \frac{p-\lowtarget}{\hightarget-\lowtarget} w^1(\beliefcycle) + \frac{\hightarget-p}{\hightarget-\lowtarget} w^0(\beliefcycle) - \kappa < \underline{w}(p)
              \end{align*}
          \end{proposition}
          The proof of this result is direct and omitted. 
          Since beliefs do not drift from $\pi$ and by \hyref{Theorem}{thm:optimal-policies} any information acquisition from $\pi$ must lead to a cycle, the dynamic problem at $\pi$ presents an effectively static choice between getting payoffs $u(\pi)$ forever or jumping directly into the best feasible cycle.
          It is also qualitatively compelling as a literal starting point since the invariant distribution $\pi$ captures the least informed an agent can be, making it a natural choice of prior.
          
        \paragraph{Symmetric problems}
        In problems invariant to relabeling the states (\emph{symmetric problems}), stationary payoffs simplify and make the mechanics of cyclical information acquisition transparent. 
        Content collapses to a one-dimensional statistic (the distance of beliefs from the long-run average) so values admit an explicit representation in terms of \emph{frequency} and \emph{quality}. 
        Symmetry is milder than it looks: common costs (entropy, variance, log-likelihood ratio) are symmetric, and solutions are invariant to affine normalizations of flow payoffs and costs.

        In symmetric problems $v(p)=v(1-p)$, implying $\lowtarget=1-\hightarget$, $\lowthreshold=1-\highthreshold$, and $\lowtime=\hightime$. 
        Cycles can rewrite as $\beliefcycle=(\targetsym,\thresholdsym,\freqsym)$ with $1/2\le\thresholdsym\le\targetsym$ and $\targetsym_{\freqsym}=\thresholdsym$: $\thresholdsym$ is the acquisition trigger uncertainty threshold, $\targetsym$ the certainty target, and $\freqsym$ the waiting time. The cyclical value then reduces to:
        \begin{align}
        \tag{SC}
        \label{eq:symmetrical-cyclical-payoffs}
        \SymmetricalCyclePayoffs(\beliefcycle)
        = \bigl(1-e^{-r\freqsym}\bigr)^{-1}\!\left(\int_0^{\freqsym} e^{-rt}\,u(\targetsym_t)\,dt - e^{-r\freqsym}\bigl(c(\targetsym)-c(\thresholdsym)-\kappa\bigr)\right).
        \end{align}
        Intuitively, each period of length $\freqsym$ accumulates discounted flow payoffs along the drift from $\targetsym$ to $\targetsym_{\freqsym}$; at the end, the DM pays $c(\targetsym)-c(\thresholdsym)-\kappa$ to reset certainty to $\targetsym$. Discounting by $e^{-r\freqsym}$ captures that period length—and hence the frequency–quality trade-off—is endogenous.

    \subsection{Comparing informativeness across environments}\label{sec:Frequency-Quality-CS}

        \paragraph{Informativeness comparisons}
        In general, there is no single or simple way to compare how much information is acquired via \emph{dynamic} processes. 
        However, the structure of solutions in this setting suggests an intuitive approach.
        There are three natural criteria that capture "more information" being acquired: static informativeness of experiments, uncertainty thresholds triggering information acquisition, and frequency.
        \begin{definition}\label{def:comparison-of-belief-cycles}
        Consider cycles $\beliefcycle=((\lowtarget,\hightarget),(\lowthreshold,\highthreshold),(\lowtime,\hightime))$ and $\tilde{\beliefcycle}=((\tilde{q}^0,\tilde{q}^1),(\tilde{p}^0,\tilde{p}^1),(\tilde{\tau}^0,\tilde{\tau}^1))$ under possibly different respective environment $(\lambda,\pi),(\tilde{\lambda},\tilde{\pi})$. Say that:
        \begin{enumerate}[label=(\roman*)]
          \item $\beliefcycle$ has more informative experiments than $\tilde{\beliefcycle}$ if: $(\tilde{q}^0,\tilde{q}^1) \subset (\lowtarget,\hightarget)$;
          \item $\beliefcycle$ has lower uncertainty thresholds for information acquisition than $\tilde{\beliefcycle}$ if: $[\tilde{p}^0,\tilde{p}^1] \subset [\lowthreshold,\highthreshold]$;
          \item $\beliefcycle$ has more frequent information acquisition than $\tilde{\beliefcycle}$ if: $\tau^i \leq \tilde{\tau}^i$ for $i=0,1$.
        \end{enumerate}
        \end{definition}
        The first notion is a mean-preserving spread condition: experiments conducted in one belief cycle are Blackwell more informative than in the other.
        These three notions induce distinct partial orders and may not agree in ranking two policies.
        In particular, note that when varying $\lambda$ or $\pi$ the frequency comparison is not implied by the comparisons in terms of beliefs. 
        This definition also illustrates that the cycle decomposition can be used to define properties tailored to applications -- e.g. when asymmetric shifts have natural interpretations.
        In symmetric problems, "how much information" is acquired condenses to three quantities ($\targetsym$, $\thresholdsym$, $\freqsym$) and each notion of informativeness becomes a complete order.
        This makes analyzing the frequency-quality tradeoff more tractable. 
        Nonetheless, there are few general comparative statics that can be obtained; rich counter-examples can be found for many expected regularities (see an illustration below for $\lambda$) even in simple cases.

        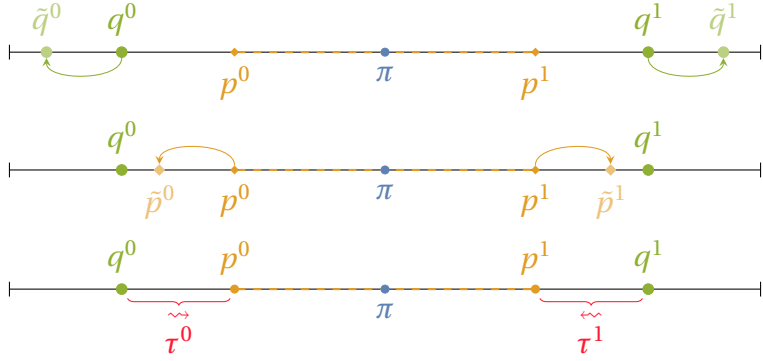
\begin{figure}[h!]
        \centering
           \begin{tikzpicture}[xscale=10,yscale=8,>=stealth]
                \draw [|-|] (0,0) -- (1,0);    

                \node[circle, draw=mathematicagreen, fill=mathematicagreen, inner sep=0pt, minimum size=4pt, label=above:{\textcolor{mathematicagreen}{$\lowtarget$}}] (pLstar) at (.15,0) {};
                \node[circle, draw=mathematicagreen, fill=mathematicagreen, inner sep=0pt, minimum size=4pt, label=above:{\textcolor{mathematicagreen}{$\hightarget$}}] (pHstar) at (.85,0) {};

                \node[diamond, draw=mathematicaorange, fill=mathematicaorange, inner sep=0pt, minimum size=3pt, label=below:{\textcolor{mathematicaorange}{$\lowthreshold$}}] (plstar) at (.3,0) {};
                \node[diamond, draw=mathematicaorange, fill=mathematicaorange, inner sep=0pt, minimum size=3pt, label=below:{\textcolor{mathematicaorange}{$\highthreshold$}}] (phstar) at (.7,0) {};

                \draw [thick, dashed, mathematicaorange] (phstar) -- (plstar); 

                \node[circle, draw=mathematicablue, fill=mathematicablue, inner sep=0pt, minimum size=3pt, label=below:{\textcolor{mathematicablue}{$\pi$}}] (pi) at (.5,0) {};

                \node[circle, draw=mathematicagreen!60!white, fill=mathematicagreen!60!white, inner sep=0pt, minimum size=4pt, label=above:{\textcolor{mathematicagreen!60!white}{$\tilde{q}^0$}}] (tildepLstar) at (.05,0) {};
                \node[circle, draw=mathematicagreen!60!white, fill=mathematicagreen!60!white, inner sep=0pt, minimum size=4pt, label=above:{\textcolor{mathematicagreen!60!white}{$\tilde{q}^1$}}] (tildepHstar) at (.95,0) {};
                \draw [->,mathematicagreen] (pLstar) edge [bend left=90] (tildepLstar);
                \draw [->,mathematicagreen] (pHstar) edge [bend left=-90] (tildepHstar);
                \end{tikzpicture}
            \\
             \begin{tikzpicture}[xscale=10,yscale=8,>=stealth]

                \draw [|-|] (0,0) -- (1,0);    

                \node[circle, draw=mathematicagreen, fill=mathematicagreen, inner sep=0pt, minimum size=4pt, label=above:{\textcolor{mathematicagreen}{$\lowtarget$}}] (pLstar) at (.15,0) {};
                \node[circle, draw=mathematicagreen, fill=mathematicagreen, inner sep=0pt, minimum size=4pt, label=above:{\textcolor{mathematicagreen}{$\hightarget$}}] (pHstar) at (.85,0) {};
        
                \node[diamond, draw=mathematicaorange, fill=mathematicaorange, inner sep=0pt, minimum size=3pt, label=below:{\textcolor{mathematicaorange}{$\lowthreshold$}}] (plstar) at (.3,0) {};
                \node[diamond, draw=mathematicaorange, fill=mathematicaorange, inner sep=0pt, minimum size=3pt, label=below:{\textcolor{mathematicaorange}{$\highthreshold$}}] (phstar) at (.7,0) {};

                \draw [thick, dashed, mathematicaorange] (phstar) -- (plstar); 

                \node[circle, draw=mathematicablue, fill=mathematicablue, inner sep=0pt, minimum size=3pt, label=below:{\textcolor{mathematicablue}{$\pi$}}] (pi) at (.5,0) {};

                \node[diamond, draw=mathematicaorange!60!white, fill=mathematicaorange!60!white, inner sep=0pt, minimum size=4pt, label=below:{\textcolor{mathematicaorange!60!white}{$\tilde{p}^0$}}] (tildeplstar) at (.2,0) {};
                \node[diamond, draw=mathematicaorange!60!white, fill=mathematicaorange!60!white, inner sep=0pt, minimum size=4pt, label=below:{\textcolor{mathematicaorange!60!white}{$\tilde{p}^1$}}] (tildephstar) at (.8,0) {};
                \draw [->,mathematicaorange] (plstar) edge [bend left=-90] (tildeplstar);
                \draw [->,mathematicaorange] (phstar) edge [bend left=90] (tildephstar);

                \end{tikzpicture}
            \\
             \begin{tikzpicture}[xscale=10,yscale=8,>=stealth]

                \draw [|-|] (0,0) -- (1,0);    

                \node[circle, draw=mathematicagreen, fill=mathematicagreen, inner sep=0pt, minimum size=4pt, label=above:{\textcolor{mathematicagreen}{$\lowtarget$}}] (pLstar) at (.15,0) {};
                \node[circle, draw=mathematicagreen, fill=mathematicagreen, inner sep=0pt, minimum size=4pt, label=above:{\textcolor{mathematicagreen}{$\hightarget$}}] (pHstar) at (.85,0) {};

                \node[circle, draw=mathematicaorange, fill=mathematicaorange, inner sep=0pt, minimum size=3pt, label=above:{\textcolor{mathematicaorange}{$\lowthreshold$}}] (plstar) at (.3,0) {};
                \node[circle, draw=mathematicaorange, fill=mathematicaorange, inner sep=0pt, minimum size=3pt, label=above:{\textcolor{mathematicaorange}{$\highthreshold$}}] (phstar) at (.7,0) {};

                \draw [thick, dashed, mathematicaorange] (phstar) -- (plstar); 

                \node[circle, draw=mathematicablue, fill=mathematicablue, inner sep=0pt, minimum size=3pt, label=below:{\textcolor{mathematicablue}{$\pi$}}] (pi) at (.5,0) {};

                \draw [decorate,decoration = {brace,raise = 3.5pt,mirror}, myred] (pLstar) --  (plstar) node [pos=.5, label=below:{\textcolor{myred}{$\overset{\rightsquigarrow}{\lowtime}$}}] {} ; 
                \draw [decorate,decoration = {brace,raise = 3.5pt,mirror}, myred] (phstar) --  (pHstar) node [pos=.5, label=below:{\textcolor{myred}{$\overset{\leftsquigarrow}{\hightime}$}}] {} ; 
                \end{tikzpicture}
            \caption{Visualizing the three notions of "more information acquisition in the long run"}
            \label{fig:information-comparison-long-run}
        \end{figure}
    
        \paragraph{Volatility and informativeness}\label{sec:CS-volatility}
        Two countervailing forces arise as volatility $\lambda$ increases.
        There is \emph{more to learn} as the environment changes faster, pushing toward more frequent tracking; information also \emph{becomes outdated faster}, so each update pays off for a shorter time, dampening tracking incentives. 
        Neither force dominates globally.
        \begin{proposition}
        \label{prop:CS-lambda-tau}
          Consider a symmetric problem. Let $\tau(\lambda)$ the time between moments of information acquisition as a function $\lambda>0$, fixing other parameters. Then:
          \begin{enumerate}
            \item $\lim_{\lambda \rightarrow 0} \tau(\lambda) = \infty$; furthermore for $\lambda$ close enough to $0$, $\tau$ is decreasing.
            \item There exists $\overline{\lambda}$ such that for all $\lambda \geq \overline \lambda$, $\tau(\lambda) = \infty$; furthermore for $\lambda$ smaller but close enough to $\overline{\lambda}$, $\tau$ is increasing.
          \end{enumerate}
        \end{proposition}
        For low levels of volatility (nearly persistent states) the need for more information dominates. 
        At the other extreme, the cost of tracking a highly volatile state becomes prohibitive: information degrades too fast, so the DM eventually saturates their capacity to profitably track the state and disengages.
        Sharper comparative statics are fragile because volatility acts through several opposing channels. 
        Fix a candidate cycle $(p,q)$: increasing $\lambda$ (i) shortens the period, raising discounted costs per unit time and shifting the path toward lower flow payoffs, yet (ii) makes cycles more frequent and lowers the “discrete” discount $1-e^{-r\tau}$, which can raise value; and (iii) endogenously shifts $(p,q)$, altering both quality and frequency. 
        Even if $p$ and $q$ move co-monotonically, this results in behavior with few restrictions -- for instance, examples can be produced where the frequency of information acquisition \emph{oscillates} as the world becomes more volatile.
        Cyclicality gives a clear descriptive handle on information acquisition, yet there is enough richness in the model to produce wide variations in observed patterns even within well-behaved classes of cost and payoff functions.

\section{Vanishing fixed costs}\label{sec:Vanishing-Fixed-Costs}

    I now study the behavior of optimal information acquisition in the limit as fixed costs vanish.
    This approximates "small" fixed cost, and also provides a tractable method and intuitive selection criterion for the limit case of continuous information acquisition.
    Intuition suggests that, as $\kappa$ goes to zero, incentives to wait for undertaking locally profitable information acquisition vanish.
    Two difficulties arise in formalizing this intuition: first, directly characterizing the limit of information acquisition \emph{policies} is challenging;\footnote{
        The nested intervals structure of solutions from \hyref{Theorem}{thm:optimal-policies} turns out not to be a tractable object as one lacks a reliable way to "track" intervals: as $\kappa$ changes, intervals in $\grossvalregionopt$ and $\inforegionopt$ could appear or disappear, merge or split in ways that are driven by the complex interplay of payoffs and costs along the belief space. 
        It is also useful to note that no claim can be made about \emph{monotonic} convergence; indeed the behavior of the belief thresholds along the path of convergence can be fairly irregular.
    } 
    second, it is a priori unclear how to define the "$\kappa=0$ limit problem" to formalize approximation arguments.
    This issue has a natural solution, deriving from a simple non-recursive reformulation of the problem, which in turn delivers tight results.
    Limit optimal policy take a simple "wait-or-confirm" form (\hyref{Theorem}{thm:optimality-and-convergence-to-wait-or-confirm}), and long run optimal dynamics have an explicit solution in terms of the concave envelope of the virtual net flow payoff (\hyref{Theorem}{thm:explicit-stationary-policy-no-fixed-cost}).

    \subsection{The limit problem}
    \label{sec:info-acquisition-with-and-without-fixed-costs}

        The key preliminary steps of the analysis consist in, first, recasting the problem as choice over \emph{belief processes}, then rewriting the resulting problem in a \emph{non-recursive} form in terms of the virtual flow payoff, which proves more suitable to consider the vanishing fixed cost limit.

        \paragraph{Belief processes}
        With no fixed cost, it may be desirable to acquire information not just at discrete points but in continuous increments.
        When allowing for continuous information acquisition, one cannot simply substitute existing analogous cost specifications based on infinitesimal information flows: those are not well defined for discrete information acquisition and vice versa.
        Nonetheless, every feasible belief process can be approximated arbitrarily well by a process where information is only acquired countably many times and costs admit an essentially unique natural extension to the whole space.
        To formalize this idea, define the set of possible belief processes, where a conditional expectation condition captures the dynamic version of Bayes plausibility with a changing state:\footnote{
          A slightly more proper but less intuitive definition can be found in \hyref{Appendix}{sec:appendix:path-measures-formalism}: it expresses the Bayes plausibility constraint in integral form as a semimartingale decomposition instead of referring to conditional expectations.
        }
        \begin{align*}
            \beliefprocesses(p) := \biggl\{ (P_t)_{t \geq 0} \text{ càdlàg process in $[0,1]$ } \biggm| \Esp[P_0]=p, \, \forall t,s \geq 0, \hspace{.1em} \Esp[P_{t+s}|P_t] \overset{a.s.}{=} e^{-\lambda s} P_t + (1-e^{-\lambda s}) \pi  \biggr\}.
        \end{align*}    
        Denote $\beliefprocessesdiscrete(p)$ the subset of belief processes that can be generated by a discrete information acquisition policy $\{\tau_i,F_i\}$ with $P_{0^-}=p$.
        The expected payoff induced by belief process $ P \in \beliefprocessesdiscrete(p)$ can be rewritten as:
        \begin{align*}
            \mathfrak{J}(P) := \Esp \Biggl[ \int_0^\infty e^{-rt} u(P_t) dt - \sum_{i \geq 0} e^{-r\tau_i} \bigl( c(P_{\tau_i^P})-c(P_{\tau_i^{P-}}) + \kappa \bigr) \Biggr].
        \end{align*}
        Notice that the uniform posterior separability of costs associated with the process $P \in \beliefprocessesdiscrete(p)$ allows this rewriting as an expectation over paths.
        From now on, denote $v_\kappa$ the value function, $w_\kappa:=v_\kappa-c$ the corresponding net value function and $P^\kappa$ the optimal belief process for $\kappa > 0$.
        By definition, $v_\kappa(p) = \max_{P \in \beliefprocessesdiscrete(p)} \mathfrak{J}(P) = \mathfrak{J}(P^\kappa)$.

        \paragraph{Virtual flow payoff reformulation}
        Recall the "virtual net flow payoff" $f$ is:
            \begin{align*}
                f(p) := u(p) - r c(p) + \lambda (\pi-p) c'(p)
            \end{align*}
        As previously, we can provide a "net" reformulation of the problem that subsumes all costs into $f$; this time, however, it is given in a non-recursive form.
        \begin{lemma}
        \label{lemma:w0-virtual-flow-form}
            The net value function $w_\kappa$ converges pointwise to $w_0$, defined by:
            \begin{align*}
                w_0(p) = \max_{P \in \beliefprocesses(p)} \Esp \Biggl[ \int_0^\infty e^{-rt} f(P_t) dt \Biggr],
            \end{align*}
            and any sequence of maximizers of $w_\kappa$ converges to a maximizer of $w_0$.
        \end{lemma} 
        \hyref{Lemma}{lemma:w0-virtual-flow-form} nests several more general results (see \hyref{Appendix}{sec:appendix:path-measures-formalism} for details). 
        The problem can generally be reformulated in the "net" form:
        \begin{align*}
            \sup_{P \in \beliefprocessesdiscrete} \mathfrak{J}(P) = \sup_{P \in \beliefprocessesdiscrete} \Esp \Biggl[ \int_0^\infty e^{-rt} f(P_t) dt - \sum_{i} e^{-r\tau_i^P}  \kappa \Biggr]
        \end{align*}
        This suggests a natural way to extend payoffs and costs to arbitrary belief processes; consistency requires establishing that discrete information acquisition can approximate arbitrarily well the \emph{performance} of any belief process.
        When $\kappa>0$, discrete information acquisition is without loss.
        When $\kappa=0$, this gives:
        \begin{align*}
            \sup_{P \in \beliefprocessesdiscrete} \Esp \Biggl[ \int_0^\infty e^{-rt} f(P_t) dt \Biggr] = \max_{P \in \beliefprocesses} \Esp \Biggl[ \int_0^\infty e^{-rt} f(P_t) dt \Biggr]
        \end{align*}
        The last step establishes that the limit of solutions as fixed costs vanish is a solution of the limit $\kappa=0$ problem.
        This is despite the fact that the expect sum of discounted fixed costs is in general \emph{not} continuous at $\kappa=0$. 
        Intuitively, costs may explode if information acquisition becomes infinitely frequent, even as $\kappa$ goes to zero, but the supremum makes such processes inadmissible.
      
        \paragraph{Remark}
        \hyref{Lemma}{lemma:w0-virtual-flow-form} equivalently shows that there is a natural extension of the primitive "discrete" costs to arbitrary belief processes -- see \hyref{Appendix}{sec:appendix:path-measures-formalism} for details.
        When the fixed cost is non-zero, continuous information acquisition will generate infinite costs, hence the initial formulation of the problem is indeed without loss.   
        In the subcase when the process features only continuous information acquisition, it can be shown the extended cost function coincides with technologies based on quantifying infinitesimal information flows -- as in for instance \cite{zhong2022optimal,hebert2023rational}. 
        \citet{georgiadisharris2023} expresses a similarly general capacity constraint over information flows for arbitrary belief martingales, under which continuous learning is optimal but discrete learning has well-defined (infinite) costs.

    \subsection{Optimal information acquisition with vanishing fixed costs}
    \label{sec:optimal-information-acquisition-with-vanishing-fixed-cost}

        When the fixed cost becomes negligible, the gap between \emph{gross} and \emph{net} interim value of information, which loosely captures the incentive to \emph{wait}, vanishes. 
        This suggests that as the fixed cost vanishes so do incentives for waiting: if there are target beliefs that the DM could profitably jump to, they should do so immediately.
        This suggests that the gap between target beliefs (the support of optimal experiments) and threshold beliefs (closest belief triggering information acquisition) would disappear.
        To formalize this idea in the language of belief processes, I introduce the class of "wait-or-confirm" policies.

        \paragraph{Wait-or-confirm policies.} 
        Informally, a "wait-or-confirm" belief process drifts until it hits the boundary of an interval in the information acquisition region, then stays at that belief until it jumps over that interval.
        Since the rate of arrival of jumps is pinned down by those two beliefs and the compensated martingale condition, it is natural to parameterize the distribution of the whole process solely in terms of the open region where information is acquired immediately -- which the belief process only possibly "jumps out of" at the initial time and never enters.

        To formalize this in a synthetic fashion, divide the boundary $\partial \mathcal{I}$ of any open set $\mathcal{I}$ into points from which drifting towards $\pi$ leads either \emph{in} or \emph{out} of the set $\mathcal{I}$:
        \begin{align*}
            \inboundary \mathcal{I} &:= \bigl\{ p \in \partial \mathcal{I} \bigm| \exists \varepsilon>0, b_\pi(p,\varepsilon) \subset \mathcal{I} \bigr\},
            \\ 
            \outboundary \mathcal{I} &:= \partial \mathcal{I} \setminus \inboundary \mathcal{I} = \bigl\{ p \in \partial \mathcal{I} \bigm| \forall \varepsilon>0, \exists q \in b_\pi(p,\varepsilon) \setminus \mathcal{I} \bigr\};
        \end{align*}
        where $b_\pi(p,\varepsilon)$ denotes the (open) "$\pi$-neighborhood" of size $\varepsilon$ at $p$: $b_\pi(p,\varepsilon):= (p,p+\varepsilon)$ if $p < \pi$ and $(p-\varepsilon,p)$ if $p > \pi$.
        \begin{definition}[Wait-or-confirm policies]
            For any open set $\mathcal{I} \subset [0,1]$ and any $p \in [0,1]$, denote $\waitorconfirmprocess_p[\mathcal{I}]$ the distribution of the belief process $P \in \beliefprocesses(p)$ such that:
            \begin{itemize}
                \item (Initial jump) If $p \in \mathcal{I}$, $P$ is distributed according to the only binary experiment in $\BayesPlausible(p)$ supported over the two closest points from $p$ not in $\mathcal{I}$, otherwise $P_0=p$ a.s.
                \item (Waiting beliefs) At all $p \in \mathcal{I}^c \cup \outboundary \mathcal{I}$, $P$ evolves according to no information acquisition, i.e it drifts deterministically with $dP_t = \lambda(\pi-P_t) dt$
                \item (Confirmation beliefs) At all $p \in \inboundary \mathcal{I}$, $P$ stays at $p$ (\emph{confirming}) until some exponentially distributed time, at which it jumps to the closest belief $q(p)$ in the direction of $\pi$ that is not in $\mathcal{I}$.
            \end{itemize}
        \end{definition}
        A belief process $P \sim \waitorconfirmprocess_p[\mathcal{I}]$ is called a \emph{wait-or-confirm (information acquisition) policy} with initial belief $p$, and $\mathcal{I}$ is called its (instantaneous) information acquisition region.
        The definition is well-posed since the open set $\mathcal{I}$ is decomposable into a countable collection of open intervals.
        At "confirmation" beliefs, the compensated martingale condition of the belief process pins down the rate of arrival of the exponential jumps to be $\lambda \frac{\pi-p}{q(p)-p}$.

        \paragraph{Optimality and convergence.} 
        The first main result in this section characterizes an optimal wait-or-confirm policy in terms of the value function $w_0$ and establishes convergence.
        \begin{theorem}
        \label{thm:optimality-and-convergence-to-wait-or-confirm}
            The net value function  $w_0$ in the problem with $\kappa=0$ is concave. Furthermore:
            \begin{enumerate}
                \item \textit{(optimality)} The wait-or-confirm process $P \sim \waitorconfirmprocess_p \bigl[ \interior L_0 \bigr]$ is optimal, where:
                    \[
                    L_0:= \biggl\{ p \in [0,1] \biggm| \exists \varepsilon>0, \exists q \in b_\pi(p,\varepsilon), w_0(q) = w_0(p) + d_\pi w_0(p) (p-q) \biggr\};
                    \]
                and $d_\pi$ denotes the directional derivative in the direction of $\pi$.
                \item \textit{(convergence)} Assume $P^\kappa \rightarrow P$, then: 
                    \[
                        P \sim \waitorconfirmprocess_p \biggl[ \liminf_{\kappa \downarrow 0} \mathcal{I}_\kappa \biggr]
                    \]
                \item \textit{(relation and local uniqueness)} It is uniquely optimal to wait to acquire information at all beliefs such that $w_0$ is strictly concave is some $\pi$-neighborhood and $\liminf_{\kappa \downarrow 0} \mathcal{I}_\kappa \subset \interior L_0$.
            \end{enumerate}
        \end{theorem}

        The proof is in \hyref{Appendix}{sec:optimal_policies_with_vanishing_fixed_costs}.
        The concavity of $w_0$ derives from there being no residual interim value of information when $\kappa=0$: if at any belief $p$ we had $\Cav[w_0](p)>w_0(p)$, then it would be optimal to immediately jump, implying $w_0(p) = \Cav[w_0](p)$; hence $w_0= \Cav[w_0]$. 
        All costs are internalized in the net value, so the expectation of $w_0$ itself is the continuation value when information is acquired.
        Hence if $w_0$ is locally strictly concave, waiting is optimal: any information acquisition would decrease payoffs by Jensen's inequality.
        Symmetrically, \emph{linear} regions capture instantaneous information acquisition being weakly profitable.
        On the boundary between such regions, it is optimal to maintain the current belief so as to "skip over" beliefs which would trigger immediate information acquisition.

        The first part of \hyref{Theorem}{thm:optimality-and-convergence-to-wait-or-confirm} constructs an optimal policy but is silent as to whether it is a unique, and convergence of optimal policies for $\kappa>0$.
        It is easy to see that neither uniqueness nor convergence to some arbitrarily selected wait-or-confirm policy can be guaranteed in general.
        Indeed, consider the case where $f$ is affine over $[0,1]$: for $\kappa=0$, every feasible belief process is optimal; for any $\kappa>0$, it is uniquely optimal to never acquire information.
        Even in this case, the optimal policy when $\kappa>0$ converges to \emph{some} wait-or-confirm policy. 
        This is what the second part of the result states formally.

        To get some intuition on the how the Poisson structure emerges, assume that "target" and "threshold" beliefs simply get closer to one another when $\kappa$ decreases to $0$, and focus on the long run interval $(\lowtarget,\hightarget)$. 
        The result implies $\lowtarget$ and $\lowthreshold$ (resp. $\hightarget$ and $\highthreshold$) should converge to the same point.
        Consider the dynamics of the beliefs starting from $\lowtarget$. 
        The wait time until the next update ($\lowtime = \frac{1}{\lambda} \log \frac{\pi-\lowtarget}{\pi-\lowthreshold}$) converges to $0$ -- i.e immediate information acquisition after a jump to $q_0$. 
        When acquiring information, the outcome is a jump:
        \begin{align*}
        \begin{dcases}
        \text{ to } \lowtarget \text{ with probability } \frac{\hightarget-\lowthreshold}{\hightarget-\lowtarget}
        \\
        \text{ to } \hightarget \text{ with probability } \frac{\lowthreshold-\lowtarget}{\hightarget-\lowtarget}.
        \end{dcases}
        \end{align*}
        Since $\lowtarget$ and $\lowthreshold$ converge to the same limit, the former probability goes to $1$ and the latter to $0$. 
        This might seem counter-intuitive: the experiment looks as if it is uninformatively confirming belief $\lowtarget$ -- but both limits (in frequency and content) are being taken simultaneously. 
        At the limiting "target belief" there is no value of information, but after an infinitesimal amount of time the inward drift of the belief triggers information acquisition which has an infinitesimally small probability of leading to a jump to the other target belief, and otherwise leads \emph{immediately back} to the previous target belief. 
        This informally describes the DM \emph{continuously checking} whether the current belief is still valid by acquiring infinitesimally informative information. 

        In other words, the underlying information technology takes the familiar form of a "Poisson breakthrough" signal. 
        The DM optimally chooses to acquire a signal structure where a breakthrough arrives at some constant rate, conditional on the true state. 
        That rate is chosen so that (i) when a breakthrough arrives, it leads to the new belief which is exactly the other target belief and (ii) the inference from the lack of arrival of the breakthrough is such that it precisely confirms the current belief i.e. cancels out the unconditional drift. 
        The same logic applies to short-run information acquisition, except that following a jump the process drifts away from the information acquisition region.
        After a "Poisson breakthrough" in the non-cyclical regime, there will be a period of no information acquisition until a new information acquisition region is hit, at which points the DM starts acquiring information again in the same fashion.

          \begin{figure}[h!]
          \centering
          \begin{subfigure}{0.4\textwidth}
              \centering
              \includegraphics[width=\textwidth]{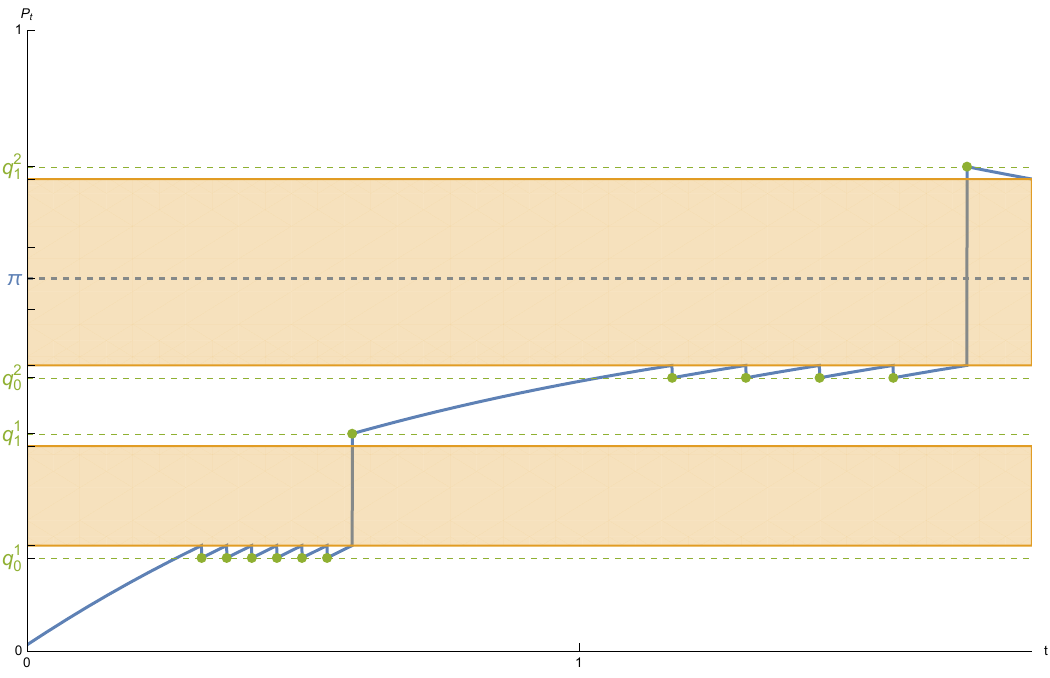}
              \caption{Belief dynamics as fixed costs become small}
          \end{subfigure}%
          ~ 
          \begin{subfigure}{0.4\textwidth}
              \centering
              \includegraphics[width=\textwidth]{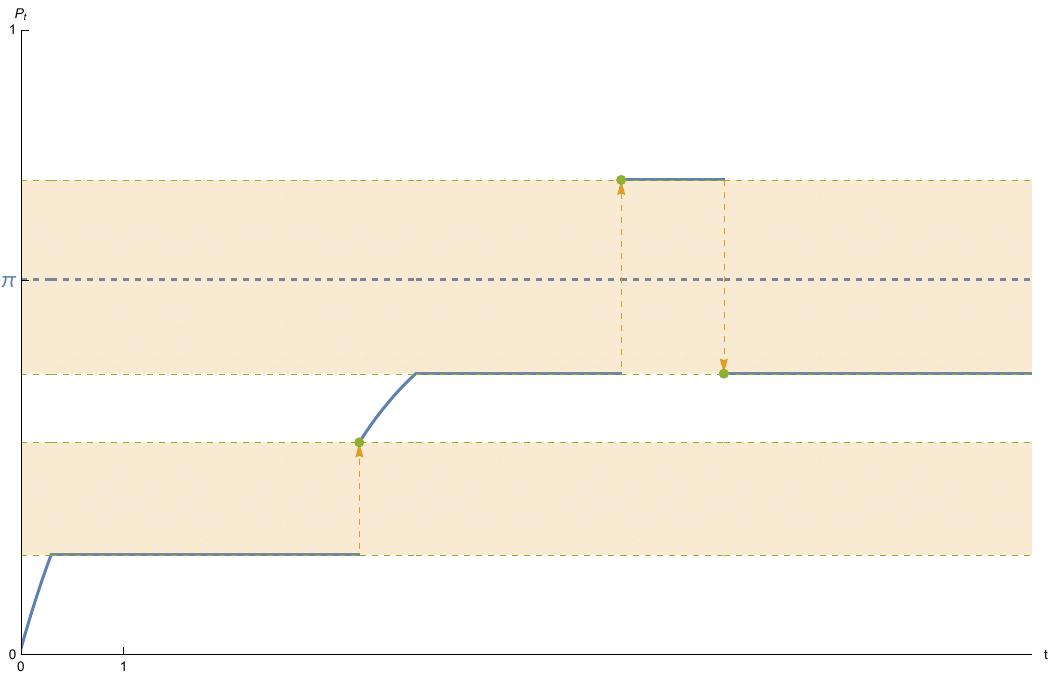}
              \caption{Belief dynamics in the limit with no fixed costs}
          \end{subfigure}
          \caption{Visualization of belief dynamics with vanishing fixed costs.}
          \label{fig:limit-process-dynamics-example}
          \end{figure}
        
        Convergence reinforces that the simple class of wait-or-confirm policies is a natural one to consider: it is without loss of optimality, consistent with continuity and limit considerations, and provides natural interpretation.
        The last point of \hyref{Theorem}{thm:optimality-and-convergence-to-wait-or-confirm} shows that the policy in the first point essentially breaks indifferences towards maximal information acquisition, whereas the limit optimal policy selects away from unnecessary information acquisition. 
        Robustness to the presence of a small fixed cost (along with selection of earliest stopping time and any selection criterion for optimal experiments) gives a natural selection of an optimal policy and that it is within the class of wait-or-confirm policies.
        Returning to the extreme illustration with $f$ linear, it is quite intuitive that when the DM is indifferent between all information acquisition policies, infinitesimal perturbation in the form of adding a vanishing fixed cost would uniquely select the policy which consists in never acquiring information.
        In that sense, the limit of optimal policies as fixed costs vanish intuitively breaks indifferences (in the limit problem) in favor of waiting, which is arguably a desirable property.

        \paragraph{Long run optimal policy.}
        In the case with no fixed costs, the regularity of the stationary regime combines with the added tractability from the virtual flow payoff reformulation in \hyref{Lemma}{lemma:w0-virtual-flow-form} to deliver an \emph{explicit} characterization of the optimal belief process in terms of the primitives of the problem.
        \begin{theorem}
        \label{thm:explicit-stationary-policy-no-fixed-cost}
            In the information acquisition problem with no fixed costs:
            \begin{itemize}
                \item If $\Cav[f](\pi)=f(\pi)$, then it is optimal in the long run to not acquire any information, uniquely so if $\Cav[f]=f$ in a neighborhood of $\pi$ and $f$ is locally strictly concave at $\pi$.
                \item If $\Cav[f](\pi)>\pi$, i.e there exists some (generically unique) interval $(q_0^f,q_1^f)$ containing $\pi$ s.t. the concave envelope of $f$ is everywhere above $f$ inside, and $f$ coincides with $Cav[f]$ at $q_0^f$ and $q_1^f$, then:
                \begin{align*}
                \forall p \in [q_0^f, q_1^f], \hspace{.5em} w_0(p) = \int_0^\infty e^{-rt} \Cav[f]\bigl( e^{-\lambda t} p + (1-e^{-\lambda t}) \pi \bigr) dt 
                \end{align*}
                and the belief process which consists in jumping from $q_0^f$ to $q_1^f$ at rate $\rho_0:= \lambda \frac{\pi-q_0^f}{q_1^f-q_0^f}$, from $q_1^f$ to $q_0^f$ at rate $\rho_1:= \lambda \frac{q_1^f-\pi}{q_1^f-q_0^f}$, and jumping immediately to $\{q_0^f,q_1^f\}$ from any $p \in (q_0^f,q_1^f)$ is optimal.
            \end{itemize}
        \end{theorem}

        The central takeaway of \hyref{Theorem}{thm:explicit-stationary-policy-no-fixed-cost} is that one can solve for optimal long run behavior by simply \emph{concavifying the virtual flow payoff $f$} around $\pi$.
        The proof strategy relies on first establishing that the given expression of the value is an upper bound for $w_0$, then showing that the feasible policy described achieves this upper bound.
        Stationary payoffs are expressed as an integral over the \emph{deterministic} drift path from $p$ when the belief process is actually random; this captures the effect that the expected payoff from jumping between $q_0^f$ and $q_1^f$ is a linear combination of the payoffs at these two points and that the probabilities move towards their long run average on the line between $f(q_0^f)$ and $f(q_1^f)$. 
        An immediate corollary of this result is that the limit of optimal long run belief cycles (and their associated belief process) for $\kappa>0$ must converge to this policy as $\kappa$ goes to $0$. 

        The explicit characterization from \hyref{Theorem}{thm:explicit-stationary-policy-no-fixed-cost} does not extend to short-run information acquisition. 
        Outside of the long run regime, the upper bound becomes strict because of transitory dynamics: it is no longer necessary that the intervals where $\Cav[f]>f$ exactly correspond to regions where it is optimal to acquire Poisson signals.
        One can still apply the idea of the recursive methodology from \hyref{Section}{sec:Dynamics}, solving for transitory optimal information acquisition within the Poisson class and "from the inside out" (constraining to one interval to the left of the stationary case, etc.), but there seems to be no explicit characterization beyond the long run.

    \subsection{Dynamics of beliefs and actions with vanishing fixed costs}
    \label{sec:dynamics-with-no-fixed-cost}

        \paragraph{Updating dynamics: random beliefs versus random times}
        As fixed costs vanish, belief changes become lumpier; information is gathered continuously in the long run.
        This derives from properties of UPS costs: there is no incentive to delay or break up profitable information acquisition.
        Time intervals where beliefs are continuously changing are the ones in which no information is acquired; learning exactly compensates the depreciation of knowledge. 
        This limit characterization reveals a transition between two regimes. 
        With positive fixed costs, information is acquired at predictable intervals, and the uncertainty lies in the outcome of the experiment.
        As fixed costs become negligible belief jumps become rare but predictable, and the uncertainty shifts to the timing of belief changes. 
        Depending on the magnitude of fixed costs, an observer might perceive the agent as either frequently and significantly adjusting their beliefs or as holding a steady belief with rare, sudden shifts. 

        \paragraph{Action switching: lumpy updating versus lumpy action}
        Because the DM eventually holds only two beliefs, they eventually take only two actions. 
        This constitutes one possible resolution of an otherwise puzzling feature: in models where each belief entails a different optimal action (e.g. a pricing problem), continuous belief drift implies continuous action adjustment. 
        This can be seen as an unpalatable prediction in general: real decision makers often only adjust their behavior lumpily.
        The literature on "sticky prices" has proposed various modeling approaches including: exogenous opportunities of changing actions \citep{Calvo1983JME,Taylor1980JPE}, menu costs \citep{SheshinskiWeiss1977RES,Mankiw1985QJE}, inattentiveness models \citep{mankiw2002sticky,MankiwReis2007JEEA,reis2006inattentiveproducers,reis2006inattentiveconsumers} which are akin to a pure fixed cost version of the present model, and rational inattention models \citep{sims2003implications,mackowiak2009optimal}.
        The limit case of the present model as fixed costs vanish outlines a different approach.
        There, discreteness derives from the flexibility of information acquisition in the absence of timing frictions in either the form of a fixed cost (like in inattentiveness models) or an exogenous time grid.
        Continuous optimal information acquisition leads the DM to optimally only ever hold finitely many beliefs; hence lumpy belief changes translate into lumpy action switches. 
        This generates new empirical questions: can one differentiate between agents who change their action periodically because they bear a cost to do so, and discrete action switches driven by information chosen so as to hold only finitely many beliefs? 
        Naturally, realistic examples might feature a mixture of explanations, with the information-driven explanation providing a complementary approach.

        \paragraph{Comparative statics} 
        The explicit long run solution in terms of $f$ makes comparative statics a much more approachable exercise.
        In particular, since analysis reduces to characterizing changes in the concave envelope of $f$ around $\pi$, one can directly use and adapt existing tools in the literature on static information acquisition and persuasion problems, see in particular \cite{curello2022comparative,whitmeyer2024}.
        This implies, in particular, that more informative beliefs are held at any time in the solution under $\tilde{f}$ than under $f$ if $\tilde{f}-f$ is a convex function.
        New interpretations follow from relating this back to primitives since $f(p)=u(p)+r c(p) -\lambda (\pi-p)c'(p)$ contains \emph{all} the parameters in the model. 
        For instance, raising the discount $r$ always imply better information acquisition in that sense; \emph{worse} information is acquired as $\lambda$ increases \emph{if} $p \mapsto (\pi-p)c'(p)$ is a convex function (this is the case with entropy costs), although the frequency of jumps increases.
    

\section{Dynamic portfolio allocation: an illustrative example}\label{sec:Examples-Applications}

    I turn to a particular application to portfolio allocation, where an underlying state governs the distribution of risky assets' return.
    Building up from a simpler case towards a richer environment, I apply the results of the previous sections and expand on their implications in this context.
    In the benchmark case, the state simply captures uncertainty about expected returns; optimal information leads to cyclical diversification with continuous portfolio rebalancing interrupted by periodic jumps to a more extreme allocation.
    Adding a friction in the form of a broker's fee generates exclusions: some initially less informed agents are confined to the safe asset by the initial cost of information acquisition.
    I then consider a case where uncertainty is not purely about relative asset returns but market conditions: all assets have higher risk and more correlated returns in one state than in the other.
    This may generates distortions between "good" and "bad" news.
    These cases are nested within a general framework that can tackle a variety of applications.

    \def\kappaval{0.02}
\def\lambdaval{0.5}
\def\pival{0.5}
\def\rval{10.0}
\def\qzero{0.053}
\def\qone{0.947}
\def\pzero{0.183}
\def\pone{0.817}
\def\tauzero{0.69}
\def\tauone{0.69}

    \pgfplotsset{
            table/search path={Tikz/Tikz_Portfolio_Example},
        }
    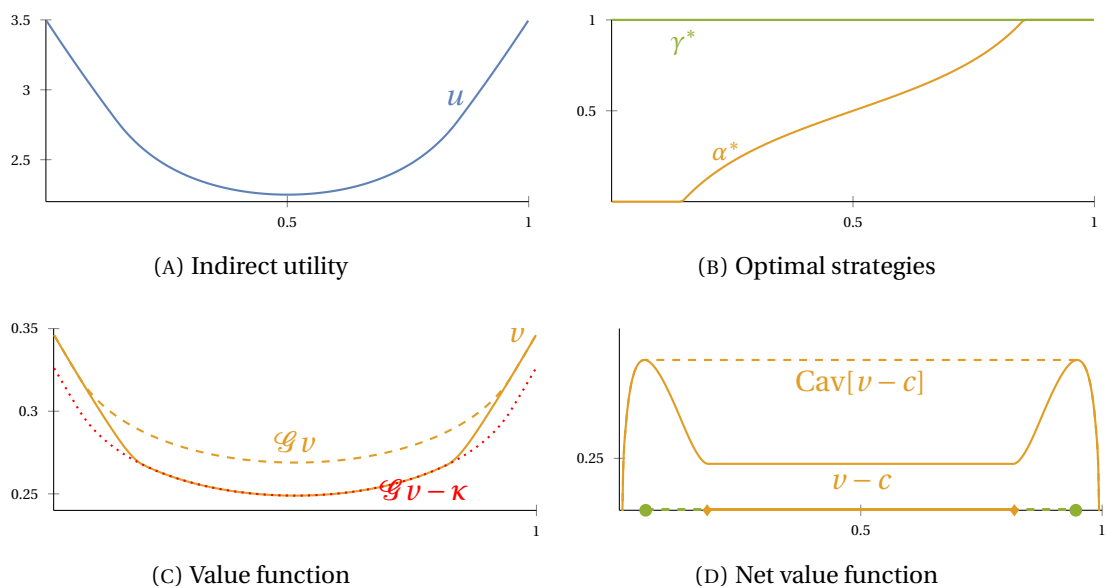
\begin{figure}[h!]
    \centering
        \begin{subfigure}{.4\textwidth}
            \begin{tikzpicture}
            \begin{axis}[
              xmin=0, xmax=1,
              ymin=2.2, ymax=3.5,  
              height=4cm, width=8cm,
              xtick = {0,.5,1},
              ytick distance = .5,
              ticklabel style = {font=\tiny},
              ylabel={},
              axis lines=middle,
              axis line style={-},
              ]
              
            \addplot+ [mathematicablue, no marks, smooth, thick] table [x=x,y=u, col sep=comma] {portfolio_example_1.csv}
              node[left,pos=.8] {$u$};

            \end{axis}
            \end{tikzpicture}
            \caption{Indirect utility}
        \end{subfigure}
        \hspace*{2em}
        \begin{subfigure}{.4\textwidth}
            \begin{tikzpicture}
            \begin{axis}[
              xmin=0, xmax=1,
              ymin=0, ymax=1,
              height=4cm, width=8cm,
              clip=false,
              xtick = {0,.5,1},
              ytick distance = .5,
              ticklabel style = {font=\tiny},
              ylabel={},
              axis lines=middle,
              axis line style={-}
              ]
              
            \addplot+ [mathematicaorange, thick, no marks] table [x=x, y=beta,  col sep=comma]{portfolio_example_1.csv} node[left,pos=.3] {\footnotesize $\alpha^*$};

            \addplot+ [mathematicagreen, thick, no marks] table [x=x, y=gamma, col sep=comma]{portfolio_example_1.csv} node[below,pos=.15] {\footnotesize $\gamma^*$};

            \end{axis}
            \end{tikzpicture}
            \caption{Optimal strategies}
        \end{subfigure}
        \\[1em]
        \begin{subfigure}{.4\textwidth}
            \begin{tikzpicture}
            \begin{axis}[
              xmin=0, xmax=1,
              ymin=.24, ymax=.35,
              height=4cm, width=8cm,
              xtick={0,1},
              ytick distance = .05,
              ticklabel style={font=\tiny},
              ylabel={},
              axis lines=middle,
              axis line style={-}
            ]

            \addplot+ [mathematicaorange, thick, no marks, smooth]
              table [x=x, y=v, col sep=comma] {portfolio_example_1.csv}
              node[left,pos=1] {$v$};

            \addplot+ [mathematicaorange, thick, dashed, no marks, smooth]
              table [x=x, y=Gv, col sep=comma] {portfolio_example_1.csv}
              node[above,pos=.5] {$\mathcal{G}v$};

            \addplot+ [red, thick, dotted, no marks, smooth]
              table [x=x, y=GvMinusKappa, col sep=comma] {portfolio_example_1.csv}
              node[right,pos=.65] {$\mathcal{G}v-\kappa$};

            \end{axis}
            \end{tikzpicture}
            \caption{Value function}
        \end{subfigure}
        \hspace*{2em}
        \begin{subfigure}{.4\textwidth}
            \begin{tikzpicture}
            \begin{axis}[
              xmin=0, xmax=1,
              ymin=.24, ymax=.275,  
              height=4cm, width=8cm,
              xtick={0,.5,1},
              ytick distance = .05,
              ticklabel style={font=\tiny},
              ylabel={},
              axis lines=middle,
              axis line style={-}
            ]

            \addplot+ [mathematicaorange, thick, no marks, smooth]
              table [x=x, y=w, col sep=comma] {portfolio_example_1.csv}
              node[below,pos=.5] {$v-c$};

            \addplot+ [mathematicaorange, thick, dashed, no marks, smooth]
              table [x=x, y=Cavw, col sep=comma] {portfolio_example_1.csv}
              node[below,pos=.5] {$\Cav[v-c]$};

            \addplot+ [mathematicagreen,mark=*,only marks,thick,
                        mark options={color=mathematicagreen,fill=mathematicagreen}] 
                table [x=x,y=targ_axis_0p24_pts, col sep=comma] {portfolio_example_1.csv};

            \addplot+ [mathematicagreen,no marks,dashed,line width=.6mm,unbounded coords=jump] 
                table [x=x,y=targ_axis_0p24_intervals, col sep=comma] {portfolio_example_1.csv};
    
            \addplot [mathematicaorange,mark=diamond*,only marks] 
                table [x=x,y=thrs_axis_0p24_pts, col sep=comma] {portfolio_example_1.csv};

            \addplot [mathematicaorange,no marks,line width=.6mm,unbounded coords=jump] 
                table [x=x,y=thrs_axis_0p24_intervals, col sep=comma] {portfolio_example_1.csv};                     
            \end{axis}
            \end{tikzpicture}
            \caption{Net value function}
        \end{subfigure}
        \caption{Solution in the portfolio choice problem, benchmark case.}
        \label{fig:portfolio-choice-benchmark}
        {\footnotesize $m_H=4,m_L=1$, $\sigma^2=2,\psi=.5,z=0$ $\kappa=\kappaval,\lambda=\lambdaval,\pi=\pival,r=\rval,c \propto \text{LLR}$ \\ $q_0=\qzero, q_1=\qone,p_0=\pzero,p_1=\pone,\tau_0=\tauzero,\tau_1=\tauone$}
    \end{figure}

    \paragraph{Dynamics of diversification} 
    An investor allocates a unit flow budget between three available assets over time; they have mean-variance preferences over returns at any date. 
    One asset is safe and yields fixed return $s \geq 0$. 
    The other two assets (labeled $A$ and $B$) are risky.
    The parameters of the distribution of returns depend on the underlying state $\theta \in \{0,1\}$, which captures fundamentals about the economy that may change over time.
    The problem is parameterized throughout in terms of the choice of a \emph{share of investment in the risky assets} $\gamma \in [0,1]$ and a \emph{subdivision of that budget} $\alpha \in [0,1]$ representing the share allocated to risky asset $A$, so that the portfolio choice problem at a given date can be written:
    \begin{align*}
    u(p) := \max_{\gamma,\alpha \in [0,1]} (1-\gamma) s + \gamma \Esp_p[ \alpha X_A + (1-\alpha) X_B ] - \frac{\psi}{2} \gamma^2 \Var_p[\alpha X_A + (1-\alpha) X_B],
    \end{align*}
    where $\Esp_p,\Var_p$ denote expectation and variance under the prior $p$ that $\theta=1$, $X_i$ is the random flow return of asset $i \in \{A,B\}$, and $\psi$ parameterizes risk aversion.     

    As a benchmark case, assume that $\theta$ captures a difference in \emph{mean return} between the risky assets -- this may, for instance, reflect the varying performance of two firms competing for a given market. 
    Formally, assume both assets have independent returns with the same variance $\sigma^2$ in either state ($\theta=0,1$), but symmetric means: $\Esp[X_A|\theta=1]=\Esp[X_B|\theta=0]=:m_H>m_L:=\Esp[X_A|\theta=0]=\Esp[X_B|\theta=1]$ ($A$ has high returns in state $1$, $B$ has high returns in state $0$).
    Denote $m_i(p):=p\Esp[X_i|\theta=1]+(1-p)\Esp[X_i|\theta=0]$ the expected return of asset $i \in \{A,B\}$ given belief $p$ that $\theta=1$, so that the problem rewrites:
    \begin{align*}
    \max_{\gamma,\alpha \in [0,1]} (1-\gamma) s + \gamma \bigl( \alpha m_A(p) + (1-\alpha) m_B(p) \bigr) - \frac{\psi}{2} \gamma^2 \bigl( \alpha^2 + (1-\alpha)^2 \bigr) \sigma^2
    \end{align*}
    Further, assume for now that $\psi < \frac{m_L - s}{\sigma^2}$ i.e. risk aversion is not too high relative to standardized excess returns relative to the safe asset. 
    This implies that (for now), it is never optimal to hold the safe asset; the optimal share of asset $A$ has closed form expression:
    \begin{align*}
     \alpha^*(p) := \frac{1}{2} \left(1 + \frac{m_A(p)-m_B(p)}{\psi \sigma^2}\right),
    \end{align*}
    which is plotted in \hyref{Figure}{fig:portfolio-choice-benchmark}. 
    Notice that $|\alpha^*(p)-1/2|$ is increasing in $|p-1/2|$: the less uncertain the DM is about the current state, the less diversified is the chosen portfolio.
    Conversely, the more uncertain they are, the more they prefer to hold a diversified portfolio.
    
    Results from the previous sections immediately imply a precise description of the investor's optimal long run information acquisition.
    Optimal behavior features a repeating pattern: continuous rebalancing of the portfolio towards greater diversification as the investor becomes more uncertain about the current state, interrupted by periodic sudden restructuring towards holding a more extreme portfolio.
    Previous work in static contexts highlighted that information acquisition may lead to under-diversification \citep[see e.g.][]{van2010information}.
    With costly periodic information acquisition in a changing environment, we see a similar effect taking place at times of information acquisition but unfolding over a cycle of endogenous length.

    \def\kappaval{0.02}
\def\lambdaval{0.5}
\def\pival{0.5}
\def\rval{10.0}
\def\qzero{0.053}
\def\qone{0.947}
\def\pzero{0.183}
\def\pone{0.817}
\def\tauzero{0.69}
\def\tauone{0.69}

    \pgfplotsset{
            table/search path={Tikz/Tikz_Portfolio_Example},
        }
    \begin{figure}
    \centering
        \begin{subfigure}{.4\textwidth}
            \begin{tikzpicture}
            \begin{axis}[
              xmin=0, xmax=1,
              ymin=0, ymax=1,  
              height=4cm, width=8cm,
              xtick = {0,.5,1},
              ytick distance = .5,
              ticklabel style = {font=\tiny},
              ylabel={},
              axis lines=middle,
              axis line style={-},
              ]
              
            \addplot+ [mathematicablue, no marks, smooth, thick] table [x=x,y=u, col sep=comma] {portfolio_example_2.csv}
              node[left,pos=.8] {$u$};

            \end{axis}
            \end{tikzpicture}
            \caption{Indirect utility}
        \end{subfigure}
        \hspace*{2em}
        \begin{subfigure}{.4\textwidth}
            \begin{tikzpicture}
            \begin{axis}[
              xmin=0, xmax=1,
              ymin=0, ymax=1,
              height=4cm, width=8cm,
              clip=false,
              xtick = {0,.5,1},
              ytick distance = .5,
              ticklabel style = {font=\tiny},
              ylabel={},
              axis lines=middle,
              axis line style={-}
              ]
              
            \addplot+ [mathematicaorange, thick, no marks] table [x=x, y=beta,  col sep=comma]{portfolio_example_2.csv} node[right,pos=.3] {\footnotesize $\alpha^*$};

            \addplot+ [mathematicagreen, thick, no marks] table [x=x, y=gamma, col sep=comma]{portfolio_example_2.csv} node[right,pos=.2] {\footnotesize $\gamma^*$};

            \end{axis}
            \end{tikzpicture}
            \caption{Optimal strategies}
        \end{subfigure}
        \\[1em]
        \begin{subfigure}{.4\textwidth}
            \begin{tikzpicture}
            \begin{axis}[
              xmin=0, xmax=1,
              ymin=0, ymax=0.09,
              height=4cm, width=8cm,
              xtick={0,1},
              ytick distance = .05,
              ticklabel style={font=\tiny},
              ylabel={},
              axis lines=middle,
              axis line style={-}
            ]

            \addplot+ [mathematicaorange, thick, no marks, smooth]
              table [x=x, y=v, col sep=comma] {portfolio_example_2.csv}
              node[left,pos=1] {$v$};

            \addplot+ [mathematicaorange, thick, dashed, no marks, smooth]
              table [x=x, y=Gv, col sep=comma] {portfolio_example_2.csv}
              node[above,pos=.5] {$\mathcal{G}v$};

            \addplot+ [red, thick, dotted, no marks, smooth]
              table [x=x, y=GvMinusKappa, col sep=comma] {portfolio_example_2.csv}
              node[right,pos=.65] {$\mathcal{G}v-\kappa$};

            \end{axis}
            \end{tikzpicture}
            \caption{Value function}
        \end{subfigure}
        \hspace*{2em}
        \begin{subfigure}{.4\textwidth}
            \begin{tikzpicture}
            \begin{axis}[
              xmin=0, xmax=1,
              ymin=-.02, ymax=.014,  
              height=4cm, width=8cm,
              xtick={0,.5,1},
              ytick distance = .01,
              ticklabel style={font=\tiny},
              ylabel={},
              axis lines=middle,
              axis x line shift=.02,
              axis line style={-}
            ]

            \addplot+ [mathematicaorange, thick, no marks, smooth]
              table [x=x, y=w, col sep=comma] {portfolio_example_2.csv}
              node[below,pos=.5] {$v-c$};

            \addplot+ [mathematicaorange, thick, dashed, no marks, smooth]
              table [x=x, y=Cavw, col sep=comma] {portfolio_example_2.csv}
              node[below,pos=.5] {$\Cav[v-c]$};

            \addplot+ [mathematicagreen,mark=*,only marks,thick,
                        mark options={color=mathematicagreen,fill=mathematicagreen}] 
                table [x=x,y=targ_axis_-0p02_pts, col sep=comma] {portfolio_example_2.csv};

            \addplot+ [mathematicagreen,no marks,dashed,line width=.6mm,unbounded coords=jump] 
                table [x=x,y=targ_axis_-0p02_intervals, col sep=comma] {portfolio_example_2.csv};
    
            \addplot [mathematicaorange,mark=diamond*,only marks] 
                table [x=x,y=thrs_axis_-0p02_pts, col sep=comma] {portfolio_example_2.csv};

            \addplot [mathematicaorange,no marks,line width=.6mm,unbounded coords=jump] 
                table [x=x,y=thrs_axis_-0p02_intervals, col sep=comma] {portfolio_example_2.csv};                     
            \end{axis}
            \end{tikzpicture}
            \caption{Net value function}
        \end{subfigure}
        \caption{Solution in the portfolio choice problem with a broker's fee.}
        \label{fig:portfolio-choice-with-brokers-fee}
        {\footnotesize $m_H=4,m_L=1$, $\sigma^2=2,\psi=.5,z=2.5$ $\kappa=\kappaval,\lambda=\lambdaval,\pi=\pival,r=\rval,c \propto \text{LLR}$ \\ $q_0=\qzero, q_1=\qone,p_0=\pzero,p_1=\pone,\tau_0=\tauzero,\tau_1=\tauone$}
    \end{figure}
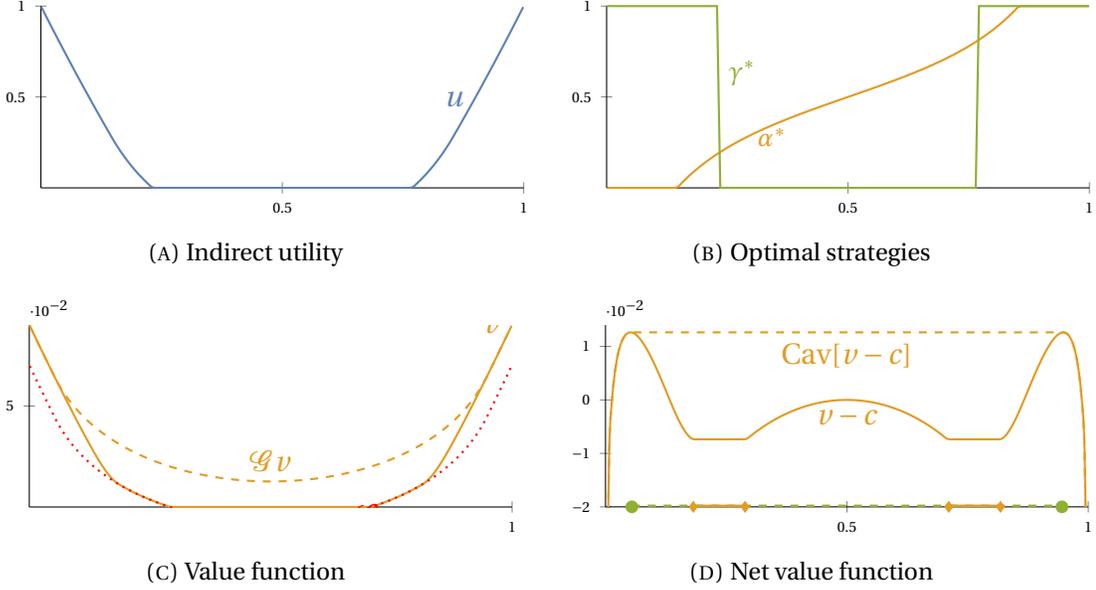
    
    \paragraph{Exclusion with frictions}
    Modify the previous problem to add a broker's fee $z>0$, which is a fixed amount paid for investing in the risky assets (\hyref{Figure}{fig:portfolio-choice-with-brokers-fee}).
    Now, the information acquisition cycle may involve phases where the investor temporarily exits the market and only holds the safe asset, especially if investment fixed costs are high or the market is intrinsically very volatile.
    The solution may also feature path dependence, with some investors being effectively excluded from the market.
    If the information costs or the fixed cost of investing are high enough, there may exist a trap region for beliefs around $\pi$. 
    If that is the case, investors who start sufficiently uninformed will never acquire any information; their belief drifts to the no information average and they hold only the safe asset.
    Meanwhile, investors who started with better information keep acquiring smaller amounts of information to maintain a cycle.
    As previously emphasized, any path dependence in the model comes from optimality -- the trap exists because initial costs are not warranted by future benefits.
    Yet, if there were externalities from information acquisition or with a concern for inequality, path dependence in access to information would have welfare implications.
    Varying parameters also alters the domain of the trap region, which provides multiple possible explanations for differences across categories of investors.

    \paragraph{Continuous monitoring} 
    As fixed costs of information acquisition become small, portfolio choice concentrates over just two possible allocations, each favoring one of the risky assets.
    Reallocation towards a more diverse portfolio vanishes in favor of sporadic but stark rebalancing.
    An investor with easier access to information adjusts their allocation less frequently, and maintains a stable investment strategy until a drastic change appears reasonable.
    The frequency and precise coarseness of periodic checks are used to calibrate the target levels of confidence the investor aims to maintain.
    By contrast, reallocation for an investor with worse access to information comes from hedging against uncertainty because of the inability to continuously monitor; this also leads them to holding a more extreme portfolio when they do update, which displays more unstable holding patterns overall.

    \paragraph{Asymmetry between market regimes}
    The state $\theta_t$ may be used to capture a more complex "market regime" which affects all assets, symmetrically or not.\footnote{
        More generally, we can consider any variation of the previous setup with arbitrarily structured returns:
        \begin{align*}
            \Esp \left(\begin{pmatrix}
                X_A \\ X_B
            \end{pmatrix}
            \middle| \theta\right)
            =
            \begin{pmatrix}
                m_A(\theta) \\ m_B(\theta)
            \end{pmatrix},
            \hspace{2em}
            \Var \left(\begin{pmatrix}
                X_A \\ X_B
            \end{pmatrix}
            \middle| \theta \right)
            =
            \begin{pmatrix}
                \sigma_A^2(\theta) & \rho(\theta)
                \\
                \rho(\theta) & \sigma_B^2(\theta)
            \end{pmatrix}
        \end{align*}
        This setup captures a variety of situations where returns jointly depend on an underlying regime and extends a standard portfolio environment to evolving market conditions in dynamic portfolio choice.
        Generalizing to more than two assets is also direct, and the analysis can also be extended beyond two states in some cases (see \hyref{Section}{sec:Discussion-Extensions}).
    }
    For instance, we could consider that $\theta_t=0$ represents a "business as usual" situation with low expected returns, low variance and low correlation -- a regime in which companies have fewer opportunities for risk-taking, so real returns are lower but safer and uncorrelated -- and $\theta_t=1$ denotes a "gold rush" situation where some exogenous disruption (e.g. a new technology) creates opportunities for high returns but entails much higher risk and correlation across assets.

    The indirect utility and optimal strategies are represented in \hyref{Figure}{fig:portfolio-choice-asymmetry}. 
    With the values chosen, the "opportunity effect" of state $1$ being higher risk and higher reward dominate, so that state $0$ being more likely is relatively bad news.
    Notice, however, that this effect is complex over the range of beliefs: the indirect utility is higher but steeper when the state is more likely to be one, and the range for which it is profitable to invest when $\theta=1$ is likelier is compressed towards more certain beliefs. 

    Breaking away from the symmetry assumption means that the optimal belief cycle need no longer be symmetric as well.
    In particular, the belief cycle may exhibit asymmetries in the frequency of information acquisition after receiving "good" or "bad" news, as well as in the relative content of what "good" and "bad" news mean in term of relative certainty.
    This brings up a potentially interesting conceptual point when linked with known biases in dynamic attention, in particular in the context of portfolio performance. 
    Some specifications of the model exhibit behavior akin to the "Ostrich effect", notably documented and analyzed in the finance literature \citep{sicherman2016financial,galai2006ostrich,karlsson2009ostrich}, which consists in agents showing bias against information acquisition after receiving bad news. 
    \citet{sicherman2016financial} study the frequency at which investors review the state of their portfolio, and two key stylized findings that are of interest in our context are that: (a) investors tend to check the status of their portfolio less frequently after poor performances than after good performances, and (b) investors tend to check less frequently when the market is more volatile. 
    While (b) could be tied to the non-monotonicity of the frequency of updates in volatility studied in \hyref{Section}{sec:CS-volatility}, (a) is characterized by patterns that can be rationalized by asymmetries within the present model (see \hyref{Figure}{fig:portfolio-choice-asymmetry}).
    Of course, this should not lead to the conclusion that this behavior is necessarily rational or purely deriving from the incentives of costly information acquisition.
    Rather, this constitutes an analytical building block which highlights that asymmetries in the frequencies of information acquisition \emph{can} arise from rationally optimal behavior; this may be useful for isolating and better understanding behavioral patterns which derive from various cognitive biases versus some form of adaptation to the environment. 

    \def\kappaval{0.02}
\def\lambdaval{0.5}
\def\pival{0.5}
\def\rval{10.0}
\def\qzero{0.075}
\def\qone{0.963}
\def\pzero{0.248}
\def\pone{0.867}
\def\tauzero{1.045}
\def\tauone{0.465}

    \pgfplotsset{
            table/search path={Tikz/Tikz_Portfolio_Example},
        }
    \begin{figure}
    \centering
        \begin{subfigure}{.4\textwidth}
            \begin{tikzpicture}
            \begin{axis}[
              xmin=0, xmax=1,
              ymin=0, ymax=1.5,  
              height=4cm, width=8cm,
              xtick = {0,.5,1},
              ytick distance = .5,
              ticklabel style = {font=\tiny},
              ylabel={},
              axis lines=middle,
              axis line style={-},
              ]
              
            \addplot+ [mathematicablue, no marks, smooth, thick] table [x=x,y=u, col sep=comma] {portfolio_example_3.csv}
              node[left,pos=.8] {$u$};

            \end{axis}
            \end{tikzpicture}
            \caption{Indirect utility}
        \end{subfigure}
        \hspace*{2em}
        \begin{subfigure}{.4\textwidth}
            \begin{tikzpicture}
            \begin{axis}[
              xmin=0, xmax=1,
              ymin=0, ymax=1,
              height=4cm, width=8cm,
              clip=false,
              xtick = {0,.5,1},
              ytick distance = .5,
              ticklabel style = {font=\tiny},
              ylabel={},
              axis lines=middle,
              axis line style={-}
              ]
              
            \addplot+ [mathematicaorange, thick, no marks] table [x=x, y=beta,  col sep=comma]{portfolio_example_3.csv} node[left,pos=.3] {\footnotesize $\alpha^*$};

            \addplot+ [mathematicagreen, thick, no marks] table [x=x, y=gamma, col sep=comma]{portfolio_example_3.csv} node[right,pos=.2] {\footnotesize $\gamma^*$};

            \end{axis}
            \end{tikzpicture}
            \caption{Optimal strategies}
        \end{subfigure}
        \\[1em]
        \begin{subfigure}{.4\textwidth}
            \begin{tikzpicture}
            \begin{axis}[
              xmin=0, xmax=1,
              ymin=0, ymax=0.15,
              height=4cm, width=8cm,
              xtick={0,1},
              ytick distance = .05,
              ticklabel style={font=\tiny},
              ylabel={},
              axis lines=middle,
              axis line style={-}
            ]

            \addplot+ [mathematicaorange, thick, no marks, smooth]
              table [x=x, y=v, col sep=comma] {portfolio_example_3.csv}
              node[left,pos=1] {$v$};

            \addplot+ [mathematicaorange, thick, dashed, no marks, smooth]
              table [x=x, y=Gv, col sep=comma] {portfolio_example_3.csv}
              node[above,pos=.5] {$\mathcal{G}v$};

            \addplot+ [red, thick, dotted, no marks, smooth]
              table [x=x, y=GvMinusKappa, col sep=comma] {portfolio_example_3.csv}
              node[right,pos=.65] {$\mathcal{G}v-\kappa$};

            \end{axis}
            \end{tikzpicture}
            \caption{Value function}
        \end{subfigure}
        \hspace*{4em}
        \begin{subfigure}{.4\textwidth}
            \begin{tikzpicture}
            \begin{axis}[
              xmin=0, xmax=1,
              ymin=-.02, ymax=.05,  
              height=4cm, width=8cm,
              xtick={0,.5,1},
              ytick distance = .01,
              ticklabel style={font=\tiny},
              ylabel={},
              axis lines=middle,
              axis x line shift=.02,
              axis line style={-}
            ]

            \addplot+ [mathematicaorange, thick, no marks, smooth]
              table [x=x, y=w, col sep=comma] {portfolio_example_3.csv}
              node[below,pos=.5] {$v-c$};

            \addplot+ [mathematicaorange, thick, dashed, no marks, smooth]
              table [x=x, y=Cavw, col sep=comma] {portfolio_example_3.csv}
              node[above,pos=.3] {$\Cav[v-c]$};

            \addplot+ [mathematicagreen,mark=*,only marks,thick,
                        mark options={color=mathematicagreen,fill=mathematicagreen}] 
                table [x=x,y=targ_axis_-0p02_pts, col sep=comma] {portfolio_example_3.csv};

            \addplot+ [mathematicagreen,no marks,dashed,line width=.6mm,unbounded coords=jump] 
                table [x=x,y=targ_axis_-0p02_intervals, col sep=comma] {portfolio_example_3.csv};
    
            \addplot [mathematicaorange,mark=diamond*,only marks] 
                table [x=x,y=thrs_axis_-0p02_pts, col sep=comma] {portfolio_example_3.csv};

            \addplot [mathematicaorange,no marks,line width=.6mm,unbounded coords=jump] 
                table [x=x,y=thrs_axis_-0p02_intervals, col sep=comma] {portfolio_example_3.csv};                     
            \end{axis}
            \end{tikzpicture}
            \caption{Net value function}
        \end{subfigure}
         \caption{Solution in the portfolio choice problem, asymmetric case.}
        \label{fig:portfolio-choice-asymmetry}
        {\footnotesize $m_A(0)=1,m_B(0)=4$, $\sigma_A^2(0)=\sigma_B^2(0)=2, \rho(0)=0$; $m_A(1)=5,m_B(1)=2$, $\sigma_A^2(1)=\sigma_B^2(1)=4, \rho(1)=.5$ \\ $\psi=.5,z=2.5,\kappa=\kappaval,\lambda=\lambdaval,\pi=\pival,r=\rval,c \propto \text{LLR}$ \\ $q_0=\qzero, q_1=\qone,p_0=\pzero,p_1=\pone,\tau_0=\tauzero,\tau_1=\tauone$}
    \end{figure}
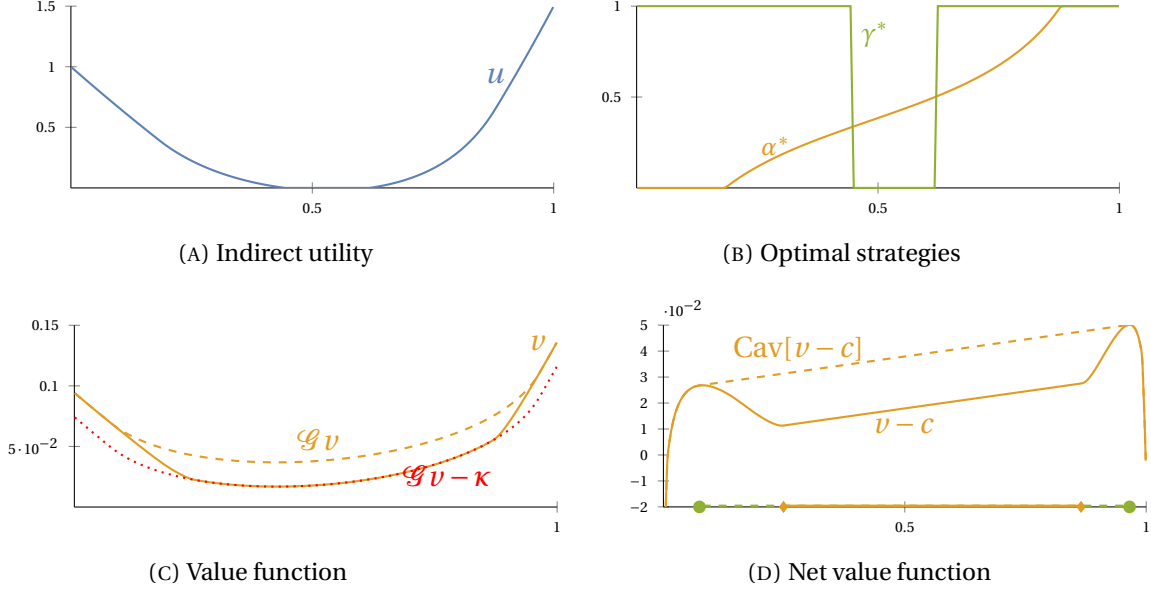

\section{Extensions and Discussion}\label{sec:Discussion-Extensions}

  \paragraph{More than two values}
    The recursive analysis is unchanged if the state has more than two values.\footnote{
        With $N$ states, transitions are governed by a time-homogeneous rate matrix $\Lambda$, which describes the rate of jumps between states; beliefs dynamics have explicit form using the matrix exponential. 
        See the appendix in the working paper version for details.}
    Optimal information acquisition regions are defined similarly and optimal experiments are constructed using supporting hyperplanes (with experiment convex polytopes replacing experiment intervals).  
    While the geometric structure of solution is preserved, belief dynamics are more complex, as belief paths can now curve through a multi-dimensional space, rather than simply drifting along a line to $\pi$. 
    In the long run, dynamics may involve "super-cycles" between multiple experiment regions rather than simple cycles. 
    The simpler cyclical dynamics generalize in some cases -- intuitively when belief paths have sufficiently low curvature (for instance if state transitions occur uniformly) or if there are sufficiently few information acquisition regions. 
    However, in most cases, the more complex dynamics introduce interesting possibilities for richer, multi-dimensional patterns of information acquisition. 

  \paragraph{Information costs}
    Generalizing the analysis to alternative information costs is challenging because the structure of optimal experiments and dynamics depends on uniform posterior separability. 
    UPS costs offer a simple belief-based framework that reduces complex dynamics to decisions over "certainty thresholds," which is useful for modeling agents with cognitive or resource constraints.
    Nonetheless, UPS costs have undesirable properties for certain applications \citep[see for example][for discussions]{denti2022posterior,caplin2022rationally,denti2022experimental,bloedel2020cost,hebert2021neighborhood}. 
    For alternative costs, the overall framework remains valid.\footnote{The recursive structure, fixed-point analysis, and general verification theorem extend to substantially weaker assumptions; results like continuity, convexity, and existence of solutions require only minimal conditions such as continuity over feasible posterior distributions and concavity in the prior \citep[as in][]{denti2022experimental}.} 
    Tractability of the continuation value operator in belief space is the essential bottleneck to explore alternative cost classes and investigate how differences in static information costs translate into dynamics of repeated information acquisition.

  \paragraph{Exogenous information}
    Results extend to accommodate exogenous background information.
    In the definition of the problem it is not essential that the belief dynamics \emph{in between moments of information acquisition} follow a deterministic drift: one may substitute richer dynamics, for instance considering a Brownian flow of information so that in between moments of information acquisition beliefs follow $dp_t = \lambda (\pi-p_t) dt + \sigma p_t(1-p_t) dB_t$.
    The recursive equation \eqref{eq:recursive-equation} is virtually unchanged; even though the stopping problem is now stochastic, the fixed point result (\hyref{Theorem}{thm:recursive-characterization}) and interval characterization of solutions (\hyref{Theorem}{thm:optimal-policies}) go through under mild regularity assumptions.
    Dynamics are naturally different but can be analyzed in the same manner: this is a simple way to incorporate richer dynamics. 
    With Brownian information acquisition times become stochastic and beliefs may escape the cycle regime.

  \paragraph{Relation to dynamic persuasion}
    The model can be reinterpreted as capturing a situation of dynamic communication, where the DM is a sender who periodically commits to sending information (at a fixed cost, but flexibly designed) about a changing state of the world to a strategic agent who takes myopically optimal decisions. 
    The equivalence is clearest in the modified Bellman equation for the net value function \eqref{eq:net-recursive-equation}: replacing the virtual flow payoff $f$ with some other arbitrary continuous function, it describes the sender's problem, mirroring results in the static problem \citep[see e.g.][]{gentzkow2014costly}. 
    This is particuarly related, formally and thematically, to \citet{ely2017beeps}: the recursive characterization mirrors his Theorem 1, with the notable difference of the fixed cost which endogenizes timing choice (versus a fixed time grid).
        
    \paragraph{Conclusion}
    The question of how to optimally adapt to an evolving environment is fundamental in a wide range of contexts.
    Because attention is a finite resource and information is costly to obtain, it is natural to expect decision makers not to constantly seek new information, and instead periodically and imperfectly update their knowledge of current circumstances.
    As a result, how frequently and how precisely decision makers acquire information has important consequences. 
    Yet, it is challenging to analyze the dynamic value of information, precisely because of the entanglement of its components across time.
    The model developed in this paper is a stepping stone, studying a tractable framework which both precisely captures the tradeoff between frequency and quality in general environments and delivers a solution method to study more detailed questions of interest.
    The general solution leads to further simplifications when focusing on specific environments. 
    Developing precise applications in e.g. finance or policy is a promising avenue for future research.
    The tractability of the model and the characterization of solution also delineates promising paths for future theoretical work integrating optimal periodic information acquisition with a changing world in more complex settings like strategic environments.

\bibliographystyle{plainnat}
\bibliography{references.bib}

\vspace{2em}

\appendix

\setcounter{section}{0}
\renewcommand{\thesection}{\Alph{section}}
\renewcommand{\thesubsection}{\Alph{section}.\arabic{subsection}}

{\huge \scshape \bfseries Appendix}

\section{Preliminaries}\label{sec:appendix:preliminaries}

    \paragraph{Belief process construction} 
    Given an initial belief $p \in \Delta(\Theta)$, arbitrary (random) sequences of increasing times $(\tau_i)_{i \geq 0}$ (with $\tau_0=0$ a.s.) and distributions $(F_i)_{i \geq 0}$ in $\Delta \Delta(\Theta)$, construct the \emph{induced belief process} $\{P_t\}_{t \geq 0}$ iteratively as follows. 
    Draw $P_0$ according to $F_0(p)$. 
    For $i \geq 0$, iterate the construction: for $t \in [0,\tau_{i+1}-\tau_i)$ the belief process is generated by the deterministic drift induced by the Markov chain, i.e. $P_{\tau_i+t} = e^{-\lambda t} P_{\tau_i} + (1-e^{-\lambda t}) \pi$ a.s. 
    If $\tau_{i+1} = \infty$, the entire belief process is characterized; otherwise if $\tau_{i+1} < \infty$, at $\tau_{i+1}$ a new belief is drawn according to the realization of the corresponding experiment, i.e. $P_{\tau_{i+1}} \sim F_{i+1}$. 
    If $\tau_\id \rightarrow T$ for some finite $T$, simply extend the process by assuming that no information acquisition occurs after $T$: $P_t = e^{(t-T) \lambda} P_{T} + (1-e^{(t-T)\lambda}) \pi$ for $t \geq T$.
    This is a well-defined construction since it can be generated using countably many uniform random variables.
    From now on, for any initial belief $p$ and sequence of times and experiments $\xi=(\tau_i,F_i)_{i \geq 0}$, denote $P^{p,\xi}=(P_t^{p,\xi})_{t \geq 0}$ the induced belief process as previously defined.

    \paragraph{Admissible policies and formal problem}
    Define $\Xi(p)$ the set of \textbf{information acquisition policies} given initial belief $p$ as the set of (random) sequences $\xi = (\tau_i,F_i)_{i \geq 0}$ of information acquisition times and experiments, such that (i) the $\tau_i$ are progressively measurable with respect to the induced belief process $P^{p,\xi}$; (ii) each experiment $F_i$ is measurable with respect to the left-limit stopped process $\bigl( \lim_{ s \uparrow t } P_s^{p,\xi}\bigr)_{t < \tau_i}$; (iii) for all $\id$, $F_\id$ is consistent with Bayes rule i.e. $F_\id \in \BayesPlausible(P^{p,\xi}_{\tau_\id^-})$. 
    For any $\xi \in \Xi(p)$, denote $(\tau_i^\xi,F_i^\xi)$ the corresponding full form. Let $\Xi = \cup_p \Xi(p)$.

    The formal statement of the information acquisition problem is:
      \begin{align*}
        v(p):= \sup_{\xi \in \Xi(p)} \Esp \Biggl[  \int_0^\infty e^{-rt} u \bigl(P^{p,\xi}_t \bigr) dt - \sum_{i \geq 0} e^{-r \tau^\xi_i} \Bigl( C(F^\xi_i)+ \kappa \Bigr) \Biggr]
      \end{align*}
    The recursive operators $\StopVal,\InfoVal$ are defined on $\CandidateValue$ as:
      \begin{align*}
        \InfoVal \tilde{v}(p) := \sup_{F \in \BayesPlausible(p)} \int \tilde{v} dF - C(F); 
        \hspace{2em} \StopVal g(p) := \sup_{\tau \geq 0} \int_0^\tau e^{-rt} u(e^{t\Lambda} p) dt + e^{-r\tau} (g(e^{\tau\Lambda} p) - \kappa)
      \end{align*}

\section{Characterization of solutions: omitted proofs}\label{sec:appendix:recursive-characterization}

    \subsection{Existence and uniqueness of a fixed point: proof of \hyref{Theorem}{thm:recursive-characterization}}\label{sec:appendix:FP-existence-uniqueness}

        It is straightforward to prove that the value function $v$ is a fixed point of $\Phi := \StopVal \circ \InfoVal$ and lies in $\CandidateValue$.
        The proof of existence and uniqueness relies on applying Theorem 1 from \cite{marinacci2019unique}.
        The space of real-valued bounded measurable functions on $\Delta(\Theta)$ equipped with the pointwise order is a Riesz space; $\CandidateValue$ is order-convex and chain-complete since on a Riesz space those sets are exactly the order intervals and $\CandidateValue = [\underline{v},\overline{v}]$. 
        Recall $\BellmanOperator= \StopVal \circ \InfoVal$ so it suffices to prove that both $\StopVal$ and $\InfoVal$ are monotone for $\Phi$ to be monotone -- which is direct using pointwise domination inside of each supremum.
        We then need to verify that $\BellmanOperator(\CandidateValue) \subseteq \CandidateValue$. Clearly $\BellmanOperator$ maps the space of bounded measurable functions to itself and for any $v \in \CandidateValue$ by definition of the supremum and feasibility of the policy which consists in never acquiring information $\BellmanOperator v \geq \underline{v}$. 
        Since $\overline{v}$ is concave and $C$ is Blackwell-monotone: $\InfoVal \overline{v}(p) = \overline{v}(p)$; hence, we have $\BellmanOperator \overline{v}(p) \leq \overline{v}(p)$, using the fact that $\overline{v}$ upper bounds the best achievable flow payoff starting from any belief in the stopping problem. 
        Using this and monotonicity of $\BellmanOperator$ gives for any $v \in \CandidateValue$, $\BellmanOperator v \leq \BellmanOperator \overline{v} < \overline{v}$. Putting the two together gives $\BellmanOperator(\CandidateValue) \subseteq \CandidateValue$.
        Both $\StopVal$ and $\InfoVal$ are order-convex using convexity of the supremum operator, so $\BellmanOperator=\StopVal \circ \InfoVal$ is order-convex. 
        The upper perimeter is given in this case by:
            \begin{align*}
              \partial^\diamond \CandidateValue = \{ w \in \CandidateValue | \inf_{p \in \Delta(\Theta)} \overline{v}(p)-w(p) = 0 \}
            \end{align*}
        i.e. it is the set of functions that get arbitrarily close to the upper bound $\overline{v}$. This is obtained by definition of the upper perimeter and can also be seen from Proposition 4 in Marinacci and Montrucchio. Now take any $w \in  \partial^\diamond I$ and any $p$ such that $\overline{v}(p)-w(p)<\varepsilon$ for $\varepsilon<\kappa/2$. We can show a direct contradiction to $\BellmanOperator w = w$ by observing that $\BellmanOperator w(p) \leq \overline{v}(p)-\kappa$ (intuitively the RHS is an upper bound on the best possible outcome: the agent cannot do strictly better than perfect observation right now and forever after, which itself cannot be obtained without paying the fixed cost at least once, even if we ignore all other costs). 
        This completes the proof of existence and uniqueness.
        Since both operators $\InfoVal,\StopVal$ map continuous functions to continuous functions, we can redo the proof of existence and uniqueness as before over the subset of continuous functions in $\CandidateValue$; since this yields a fixed point of $\BellmanOperator$ in a subset of $\CandidateValue$, this fixed point must be the unique one over the whole set.
   
    \subsection{Optimal policies: proof of \hyref{Theorem}{thm:optimal-policies}}\label{sec:appendix:optimal-policies}

        \paragraph{Optimal experiments and the continuation value operator} 
        Implications of uniformly posterior separable costs for \emph{static} information acquisition are now well-studied. 
        The following proposition states a collection of key properties in the context of this model. 
        These can be proven using, for instance, results in \cite{caplin2022rationally,gentzkow2014costly,dworczak2024persuasion}.

        \begin{appendixproposition}
        \label{prop:concavification}
            Consider an arbitrary continuous function $\tilde{v}$ over $\Delta(\Theta)$.
            \begin{enumerate}[label=(\roman*)]
              \item \textbf{Value of information via concave envelope:} $\InfoVal \tilde{v}(p) = \Cav [\tilde{v}-c] (p) + c(p)$.
              \item \textbf{Geometric characterization of optimal experiments:} any optimal experiment at $p \in \Delta(\Theta)$ is supported over points where the supporting hyperplane of the convex hull of the subgraph of $\tilde{v}-c$ at $p$ meets the graph of $\tilde{v}-c$ and conversely, any Bayes-Plausible experiment supported over those points is optimal. Further, for any $p \in \Delta(\Theta)$ there must exist some optimal experiment which induces at most $|\Theta|$ possible distinct posteriors (i.e. $|\supp(F_p)| \leq |\Theta|$).
            \end{enumerate}
        \end{appendixproposition}


        \paragraph{Stopping value and dynamic value of information acquisition.} 
        The following proposition restates classical facts about optimal stopping (see e.g. \cite{peskir2006optimal}).
        \begin{appendixproposition}
        \label{prop:optimal-stopping-facts}
            Let $w$ a candidate value function. The value in the problem of optimal timing of one-shot information acquisition with continuation value $w$, which is given by $\BellmanOperator w = \StopVal (\InfoVal w)$, is the unique solution $\tilde{w}$ to:
            \begin{align*}
                \min \biggl\{ r \tilde{w}(p) - u(p) - \lambda (\pi-p) \tilde{w}'(p), \tilde{w}(p) - \InfoVal w(p) \biggr\} = 0.
            \end{align*}
            The stopping policy defined by $\tau(p) := \inf \biggl\{ t \geq 0 \biggm| \BellmanOperator w(p_t) = \InfoVal w(p_t) \biggr\}$ is optimal.
        \end{appendixproposition}

        \paragraph{A general verification theorem} 
        Combining \hyref{Proposition}{prop:concavification} and \hyref{Proposition}{prop:optimal-stopping-facts} above to give \hyref{Theorem}{thm:optimal-policies} is justified by the following verification result, which is fairly direct.
        \begin{appendixproposition}[Verification: optimal policies given value function]
        \label{prop:verification-theorem-general-form}
            Let $v$ be the unique fixed point of $\BellmanOperator$. Any optimal strategy $\{\tau_i,F_i\} \in \Xi$ must verify a.s. for any $\id \in \naturals$:
            \begin{align*}
                \tau_i - \tau_{i-1} & \in \argmax_{\substack{\tau \in [0,\infty] \\ p_0 = P_{\tau_{i-1}}}} \int_0^\tau e^{-rt} u(p_t) dt + e^{-r\tau} \InfoVal v(p_\tau);
                \\[.5em]
                F_i & \in \argmax_{F \in \BayesPlausible(P_{\tau_i^-})} \int v dF - C(F).
            \end{align*}
            Where both argmaxes are non-empty a.s. Conversely, any strategy which is almost surely induced in this way by iterated selections of measurable mappings is optimal.
      \end{appendixproposition}


    \subsection{Convergence: proof of \hyref{Theorem}{thm:long-run-dynamics}}\label{sec:appendix:dynamics-long-run-behavior}

        Define the "long run domain" $D=[q^0,q^1]$ as either the (closure of the) interval in $\grossvalregionopt$ that contains $\pi$, if there is one, or if not an arbitrarily chosen closed interval around $\pi$ in which no information is acquired.
        Once the belief process enters $D$, it must either follow a cycle or information acquisition must stop.
        Hence it suffices to prove that the first entry time of the process in $D$ is almost surely finite.

        Consider an arbitrary initial belief $p$, assume without loss that $p<\pi$ (the proof is symmetric in the alternative case) and that $p$ is in the waiting region (the initial jump makes no difference).
        Denote by $\{(q^0_n,p^0_n,p^1_n,q^1_n)\}_{n \in N}$ the collection of "effective on-path information acquisition intervals", i.e intervals in $\grossvalregionopt$ such that $\grossvalregionopt \cap \inforegionopt \neq \emptyset$ (where $(q^0_n,q^1_n)$ denote the endpoints of the interval in $\grossvalregionopt$ and $(p^0_n,p^1_n)$ the minimum and maximum of $\grossvalregionopt \cap \inforegionopt$ respectively) that are between $p$ and $\pi$ but not in $D$.  Note $N$ is countable but not necessarily finite -- label intervals using the natural numbers in a natural ordered fashion from left to right, i.e. $p < p^0_0$ and for all $n$ $q^0_n < q^1_n < q_0^{n+1}$.
        For any $\tilde{q}<\tilde{p}<\pi$, denote $\tau(\tilde{q},\tilde{p})$ the time it takes for beliefs to deterministically drift from $\tilde{q}$ to $\tilde{p}$.
        Define the following sequence of independent random variables:
        \begin{align*}
            T_0 & := \tau(p,p^0_0) + \tau(q^0_0,p^0_0) \times X_0 \text{ where } X_0 \sim \mathcal{G}\biggl( \frac{p_0^0-q_0^0}{q^1_0-q^0_0} \biggr)
            \\
            \forall n \geq 1, \; & T_n := \tau(q_1^{n-1},p^0_{n-1}) + \tau(q^0_n,p^0_n) \times X_n \text{ where } X_n \sim \mathcal{G}\biggl( \frac{p_n^0-q_n^0}{q^1_n-q^0_n} \biggr)
        \end{align*}
        where all the $X_n$ are independent and defined on some probability space $(\Omega,\mathcal{F},\prob)$, and $\mathcal{G}$ denotes the geometric distribution.
        $T_n$ describes the amount of times it takes to "cross over" the $n$-th interval of information acquisition, after crossing the $n-1$-th.     
        The main object of interest is the \emph{total} time it takes to cross over all effective on-path information acquisition intervals, which is $\mathbf{T}:= \sum_{n \in N} T_n$.
        The event $\{ \mathbf{T} = \infty \}$ is a tail event in the sense that it is in the terminal $\sigma$-algebra of the sequence of $\sigma$-algebras generated by the $T_n$, hence a classical application of Kolmogorov's 0-1 law entails that $\prob (\mathbf{T} = \infty) \in \{0,1\}$.
        Denote $E := \{ \omega \in \Omega | \forall n, \; X_n = 1 \}$ the set of realizations such that the process jumps over each interval on the first information acquisition time. 
        By definition of the $X_n$, $\prob(E)>0$ and by construction of the $T_n$, for any $\omega \in E$: $\mathbf{T}(\omega) \leq \tau(p,q^0) < \infty$.
        Hence $\prob (\mathbf{T} < \infty) > 0$, so it must be that $\prob (\mathbf{T} < \infty) = 1$. 
        Up to a constant, $\mathbf{T}$ is the first entry time of the belief process in $D$, so this completes the proof.

\section{Information acquisition with and without fixed cost}\label{sec:appendix:path-measures-formalism}
    
    \subsection{Optimization over arbitrary belief processes -- proof of \hyref{Lemma}{lemma:w0-virtual-flow-form}}

        \paragraph{Belief processes} 
        Throughout, let $D$ the space of $\Delta(\Theta)$-valued càdlàg functions over $[0,\infty)$; since the domain will always be time we call elements of $D$ "belief paths". 
        Equip $D$ with the usual Skorohod metric, denoted $d$ (see \cite{billingsley2013convergence} for a formal definition), making it a complete separable metric space. 
        Further equip $D$ with the Borel $\sigma$-algebra induced by its metric topology and denote $\Delta(D)$ the set of probability measures over this measurable space.
        
        Let $(\Omega,\mathcal{F},\mathbb{Q})$ a probability space. 
        A belief process is viewed as a random variable over $D$, i.e. a measurable function from $(\Omega,\mathcal{F},\mathbb{Q})$ to $D$ equipped with the Borel $\sigma$-algebra generated by its metric topology induced by $d$, with a compensated martingale property. 
        Let $\beliefprocesses(p)$ the set of belief processes, given initial belief $p$:
        \begin{align*}
            \beliefprocesses(p):= \biggl\{ P & \text{ càdlàg process in } [0,1] \biggm| \Esp[P_0]=p, 
            \\ & \quad \exists (M_t) \text{ a local martingale such that } P_t \overset{a.s}{=} P_0 + M_t + \int_0^t \lambda (\pi - P_s) ds \biggr\}
        \end{align*}
        Identify elements in $\beliefprocesses(p)$ with their distributions and slightly abuse notations to interpret $\beliefprocesses(p)$ either as a space of random variables or as a space of measures.
        Further recall that for any information acquisition policy $\xi \in \Xi$, $P^{\xi}$ denotes the corresponding belief process (omitting dependence on the initial point). 
        Denote by $\beliefprocessesdiscrete(p)$ the subset of $\beliefprocesses(p)$ such that information is only acquired at countably many moments in time: $\beliefprocessesdiscrete(p) = \biggl\{ Q \in \beliefprocesses(p) \biggm| \exists \xi \in \Xi, \hspace{.5em} Q = P^{\xi,p} \text{ a.s.} \biggr\}$.
        $\beliefprocessesdiscrete(p)$ is a dense class in $\beliefprocesses(p)$ (it contains in particular any discrete-time approximation on a grid).
        For any belief process in $P \in \beliefprocessesdiscrete$, denote $\{\tau_i^P\}_{i \in \naturals}$ the ordered random times of the discontinuities of $P$ and for any $i$ let $F_i^P$ the random distribution of $P_{\tau_i}$.
        Straightforwardly $\{\tau_i^P,F_i^P\} \in \Xi$ and this pins down a unique (in distribution) element of $\Xi$ so that there is a one to one mapping between $\beliefprocessesdiscrete$ and $\Xi$.

        \paragraph{Payoffs over belief processes}
        The information acquisition problem for $\kappa>0$ rewrites as:
        \begin{align*}
        v_\kappa(p) := \max_{P \in \beliefprocessesdiscrete(p)} \Esp \left[ \int_0^\infty e^{-rt} u(P_t)dt - \sum_{i \geq 0} e^{-r\tau_i^P} \left( \Esp \left[c \left(P_{\tau_i^P} \right) - c\left(P_{\tau_i^{P-}} \right) \middle| P_{\tau_i^{P-}} \right] + \kappa \right) \right].
        \end{align*}
        The discounted utility term is consistently defined for all belief processes in $\beliefprocesses(p)$.
        For any $p \in D$, define $\Upath(p) := \int_0^\infty e^{-rt} u(p_t) dt$, which is a continuous function over the space of paths: for any sequence of paths $p_n \in D$, $p_n \xrightarrow[]{d} p$ implies that for some sequence of increasing continuous bijective function $\lambda_n : \reals_+ \rightarrow \reals_+$ such that $\lambda_n(t) \rightarrow t$ for all $t$ as $n \rightarrow \infty$, $p_n(\lambda_n(t)) \rightarrow p(t)$ for almost every $t$ (since every càdlàg path has countably many discontinuities, see Lemma 5.1. in \cite{ethier2009markov}); using convergence almost everywhere and continuity of $u$ with a dominated convergence argument over the bound $|u(p_n(t))-u(p(t))| \leq |u(p_n(t))-u(p_n(\lambda_n t))|+|u(p_n(\lambda_n t))-u(p(t))|$, we get that $U(p_n) \rightarrow U(p)$. 
        Utility over belief processes is simply expected utility, where the expectation is over paths: $\Uprocess(\mu) := \int_{D} \Upath d\mu$, which is well-defined and continuous since $u$ is bounded and continuous.
        
        The cost term is only defined over belief processes in $\beliefprocessesdiscrete(p)$, i.e. such that all randomness occurs in countably many jumps. 
        Since $\beliefprocessesdiscrete(p)$ is dense in $\beliefprocesses(p)$, it is natural to define costs over the whole class via the limit of costs from approximations in $\beliefprocessesdiscrete(p)$.
        However, it is a priori unclear whether costs satisfy uniform continuity properties required to make this extension generally well-defined; to make this formal, it is helpful to first rewrite the cost.
        Consider an arbitrary belief process $P \in \beliefprocessesdiscrete$ with discontinuities at times $(\tau_i)_{i \geq 0}$.
        Consider first the case $\tau_0 > 0$; rewrite, with all equalities pathwise:
        \begin{align*}
            \sum_{ i \geq 0 } e^{-r\tau_i} \biggl( c(P_{\tau_i}) - c(P_{\tau_i^-}) \biggr) 
            & = - \sum_{ i \geq 0} \biggl( e^{-r\tau_i} c(P_{\tau_i^-}) - e^{-r\tau_{i-1}} c(P_{\tau_{i-1}}) \biggr) - c(p) \quad \text{ with $\tau_{i-1}:=0$ }
            \\
            & =  - \sum_{ i \geq 0 } \int_{\tau_{i-1}}^{\tau_i} \frac{\partial}{\partial t} \bigl( e^{-r t} c(P_t) \bigr) dt - c(p) 
            \\
            & = \int_0^\infty e^{-rt} \bigr( r c(P_t) - \lambda (\pi - P_t) c'(P_t) \bigr) dt - c(p) 
        \end{align*}
        Where the last equality uses that $dP_t = \lambda(\pi-P_t)dt$ between the $\tau_i$ (hence almost everywhere); this can be equivalently obtained by writing explicitly $P_{\tau_i^-} = e^{-\lambda(\tau_i-\tau_{i-1})} P_{\tau_{i-1}} + \bigl(1-e^{-\lambda(\tau_i-\tau_{i-1})} \bigr)\pi$.
        If $\tau_0=0$, the proof is the same with sums starting from $i=1$.

        To extend the \emph{variable cost} component to arbitrary belief processes, define for any $P \in \beliefprocesses(p)$:
        \begin{align*}
         \Cprocess(P):= \Esp \left[ \int_0^\infty e^{-rt} \bigr( r c(P_t) - \lambda (\pi - P_t) c'(P_t) \bigr) dt - c(p) \right]
        \end{align*}
        This is always well-defined with possibly infinite costs $\Cprocess(P) \in[0,\infty]$. 
        Indeed, a direct application of the Meyer-It\^o formula for convex functions of semimartingales, with $dP_t = \lambda (\pi - P_t) dt + dM_t$ for some martingale $M$, gives $\Esp[c(P_t)] \geq c(p) + \Esp \left[ \int_0^t \lambda (\pi - P_s) c'(P_s) ds \right]$.
        Further since $q \mapsto r c(q) - \lambda (\pi - q) c'(q)$ is bounded below by convexity of $c$, we can always write its integral along the path $P_t$ (it may take value $+\infty$). 
        If the integral explodes with positive probability there is nothing further to prove. 
        If the integral converges then integration by part shows that the process $X_t := e^{-r c(P_t)} + \int_0^t e^{-rs} \bigl( rc(P_s) - \lambda (\pi - P_s) c'(P_s) \bigr)$ is a submartingale; passing to the limit yields $\Cprocess(P) \geq 0$.
        Furthermore, this function is continuous over the domain of \emph{admissible} processes (i.e. with finite costs). 
        If $c$ and $c'$ are bounded, this is all belief processes; if $c$ or $c'$ equals infinity at a boundary point, then it is without loss of optimality to restrict attention to processes that belong to a restricted compact domain a.s. after some initial amount of time; we recover continuity in the Skorohod topology since $q \mapsto r c(q) - \lambda (\pi - q) c'(q)$ is continuous over the restricted domain.
    
        To deal with the \emph{fixed cost} component, observe that for all $P \in \beliefprocesses(p) \setminus \beliefprocessesdiscrete(p)$ and any approximating sequence $P^n \in \beliefprocessesdiscrete(p)$ with $P^n \rightarrow P$, letting $\tau_i^n$ the jump times of $P^ n$ we have: $\Esp \left[ \sum_{i} e^{-r\tau_i^n} \kappa  \right] \rightarrow +\infty$. 
        This is intuitive because approximation of a process which is \emph{not} in $\beliefprocessesdiscrete(p)$ must require arbitrary close jump times (an infinitely fine time grid), so discounted fixed costs explode. 
        Therefore we can "continuously" extend the fixed cost component by letting for any $P \in \beliefprocesses(p)$:
        \begin{align*}
            \mathfrak{H}_\kappa(P) := \begin{cases}
                \sum_{i} e^{-r \tau_i^P} \kappa & \text{ if } P \in \beliefprocessesdiscrete(p)
                \\ 
                + \infty & \text{ otherwise }
            \end{cases}
        \end{align*}
        Putting everything together, we have that by construction for any $\kappa \geq 0$:
        \begin{align*}
            \sup_{P \in \beliefprocessesdiscrete(p)} \Uprocess(P) - \Cprocess(P) - \mathfrak{H}_\kappa(P) = \sup_{P \in \beliefprocesses(p)} \Uprocess(P) - \Cprocess(P) - \mathfrak{H}_\kappa(P)
        \end{align*}
        For $\kappa>0$, this is exactly our original problem: nothing is added since non-discrete information acquisition bears infinite costs. 
        When $\kappa=0$, this extends the problem where the fixed cost term disappears ($\mathfrak{H}_0 \equiv 0$): the objective function $\Uprocess-\Cprocess$ is now well-defined for any belief process in $\beliefprocesses(p)$; density of $\beliefprocessesdiscrete(p)$ and continuity up to infinitely costly processes guarantees that the payoffs of any candidate belief process can be approximated arbitrarily close by processes in $\beliefprocessesdiscrete(p)$, yielding the equality.
        Lastly, the "net" form of the problem comes directly from rewriting $\Uprocess-\Cprocess$ as a single integral and taking the constant $c(p)$ to the other side:
        \begin{align*}
            w_\kappa(p) := v_\kappa(p) - c(p) = \sup_{P \in \beliefprocesses(p)} \int e^{-rt} \underbrace{(u(P_t)-rc(P_t)+\lambda(\pi-P_t)c'(P_t))}_{=f(P_t)}dt - \mathfrak{H}_\kappa(P)
        \end{align*}
        
        \paragraph{Limit of solutions and solution of the limit problem} 
        The cost function is not continuous in the parameter $\kappa$ at $0$, so convergence of solutions to a solution of the limit problem as $\kappa$ vanishes cannot be proven with e.g. Berge's Theorem of the maximum.
        The notion of epi-convergence, which is the weakest notion of functional convergence which guarantees convergences of minimizers, is better suited. 
        First recall the definition, and its main implication, the proof of which is direct \citep[see][for classical references]{BeerRockafellarWets1992,Attouch1984,AttouchWets1989,RockafellarWets1998}.
        \begin{appendixdefinition}
            Let $(X,d)$ a metric space and a sequence of functionals $f_n : X \rightarrow \overline{\reals}$. Say that $f_n$ epi-converges to $f : X \rightarrow \overline{\reals}$ and denote it $f_n \xrightarrow[n \rightarrow \infty]{\text{epi}} f$ if for every $x \in X$:
            \begin{enumerate}[label=(\roman*)]
                \item For any $x_n$ s.t. $x_n \rightarrow x$, $f(x) \leq \liminf_{n} f_n(x_n)$,
                \item There exists $x_n \rightarrow x$ such that $f(x) \geq \limsup_n f_n(x_n)$.
            \end{enumerate}
        \end{appendixdefinition}
        \begin{appendixproposition}
            If $f_n \xrightarrow[n \rightarrow \infty]{\text{epi}} f$ then for any sequence $x_n \in \argmin f_n$:
            \begin{align*}
                x_n \rightarrow x \Longrightarrow x \in \argmin f
            \end{align*}
        \end{appendixproposition}
        \begin{appendixlemma}
            Let $\kappa_n$ a sequence of strictly positive real numbers such that $\kappa_n \rightarrow 0$ as $n \rightarrow \infty$, then $\mathfrak{H}_{\kappa_n}$ epi-converges to the constant $0$ function as $n$ goes to infinity. 
        \end{appendixlemma}

        \begin{proof}
            Consider some arbitrary $P \in \beliefprocesses(p)$. 
            By definition for any $P^n \rightarrow P$, $\liminf_n \mathfrak{H}_{\kappa_n}(P^n) \geq 0 = \mathfrak{H}_{0}(P)$. 
            To show the second condition, construct a sequence $P^n$ as follows: fix a time grid with uniform step size $h_n$ -- let $\tau_0=0$, $\tau_j = \tau_{j-1}+h_n$; let $P^n$ follow the unconditional drift for each $t \in [\tau_j,\tau_{j+1})$ and "update" the process to $P$ at each $\tau_j$: $P^n_{\tau_j}=P_{\tau_j}$ for all $j$. 
            By construction of $P^n$ discontinuities are at the fixed $h_n$-spaced grid points, so that:
                \begin{align*}
                    \mathfrak{H}_{\kappa_n}(P^n) = \sum_{j \in \naturals} \kappa_n e^{-r h_n j} = \frac{\kappa_n}{1- e^{-r h_n}},
                \end{align*}
            choosing grid steps such that $h_n \rightarrow 0$ and $\frac{\kappa_n}{1-e^{-r h_n}} \rightarrow 0$ gives that $\mathfrak{H}_{\kappa_n}(P^n) \rightarrow 0$.
        \end{proof}
        
       This ensures that every limit of solutions of the problem with a fixed cost goes to a solution of the problems with no fixed cost as $\kappa$ vanishes, as the following corollary states formally.
        \begin{appendixcorollary}
        \label{prop:convergence-of-solutions-general}
            Let $v_\kappa$ the value function corresponding to any $\kappa \geq 0$ and $P^\kappa$ the optimal belief process. Then for any $\overline{\kappa} \geq 0$: (i) $v_\kappa$ converges pointwise to $v_{\overline{\kappa}}$ as $\kappa \rightarrow \overline{\kappa}$; (ii) if $P^\kappa$ converges in distribution to $P^{\overline{\kappa}}$ as $\kappa \rightarrow \overline{\kappa}$, then $P^{\overline{\kappa}}$ is an optimal belief process in the problem with fixed cost $\overline{\kappa}$.
        \end{appendixcorollary}

        \paragraph{Relation to infinitesimal information flows}
        The "extended" variable cost function considered over processes that feature only continuous information acquisition coincides with (a natural extension to the changing state environment of) existing cost functions based on infinitesimal information flows -- see e.g. \cite{zhong2022optimal,bloedel2020cost,hebert2023rational}. 
        Indeed, assume $P$ is a belief process and let $M$ its martingale component (the "information acquisition" part of the change in beliefs). 
        If $M$ has a well-defined infinitesimal generator $\mathcal{L}_M$ over some appropriate domain of functions, then the same holds for $P$ and Dynkin's formula for the discounted cost process gives:
        \begin{align*}
         \Esp[e^{-rT} c(P_T) ] = c(p) + \Esp \Biggl[ \int_0^T e^{-rs} \biggl( \mathcal{L}_M c(P_t) - r c(P_t) + \lambda (\pi - P_t)c'(P_t) \biggr) dt \Biggr],
        \end{align*} 
        where existence of the infinitesimal generator now imposes $P_0=p$ (no deterministic jumps).
        Assuming costs are finite, rearranging and taking a limit yields:
        \begin{align*}
            \Cprocess(P) = \Esp \Biggl[ \int_0^\infty e^{-rt} \mathcal{L}_M c(P_t) dt \Biggr]
        \end{align*}
        Belief processes generated from either only discrete or only continuous information acquisition do not overlap and do not cover the whole space of belief processes. 
        Defining general costs via approximating any belief process arbitrarily well with discrete information acquisition consistenly extends these ideas over arbitrary processes while bypassing the question of existence of the generator or optimality of continuous information acquisition.
        A similar point is made in \citet{georgiadisharris2023} with a capacity constraint: the structure of optimization over belief processes naturally allows for a general semimartingale formalism, which specializes to expressions in terms of the characteristics when they are well-behaved.

    \subsection{Optimal policies with vanishing fixed costs}
        \label{sec:optimal_policies_with_vanishing_fixed_costs}

        \paragraph{Proof of \hyref{Theorem}{thm:optimality-and-convergence-to-wait-or-confirm}}
        The proof is separated into two results, which are proven below.
        \begin{appendixproposition}
        \label{prop:optimal-jump-policy-from-w0}
            The optimal net value function  $w_0$  in the information acquisition problem with $\kappa=0$ is concave and: (i) for every belief $p$ such that $w_0$ is strictly concave in a $\pi-$neighborhood $p$, it is uniquely optimal to not acquire information at $p$; (ii) for every belief $p$ such that $w_0$ is locally affine at $p$, it is optimal to immediately acquire information at $p$; (iii) for every belief $p$ such that neither previous condition hold, it is optimal to acquire information so as to \emph{confirm} $p$ until some exponentially distributed time, at which beliefs jump to some prescribed belief $q(p)$ in the direction of $\pi$.
        \end{appendixproposition}
        \begin{proof}
            If $w$ is strictly concave locally in the direction of $\pi$ at $p$, then it must be that the optimal process has $P_0 = p$ a.s. since otherwise we would have $\Esp[w(P_0)]<w(p)$, as any feasible belief process must put positive probability in the direction of $\pi$. 
            If this is true for any $\tilde{p}$ in some $\pi$-neighborhood of $p$, then the only possible belief process is one that follows the deterministic drift $dP_t = \lambda (\pi-P_t) dt$ in that $\pi$-neighborhood.
            If instead $w$ is locally affine in some neighborhood $[q_0,q_1]$ of $p$, then clearly:
            \begin{align*}
             \frac{p-q_0}{q_1-q_0} w(q_1) + \frac{q_1-p}{q_1-q_0} w(q_0) = w(p),
             \end{align*}
            so it is optimal to immediately acquire information so as to jump to $\{q_0,q_1\}$.
            Now consider any $p$ such that $w$ is not locally affine at $p$ but $w$ is also not strictly concave in any $\pi$-neighborhood of $p$. 
            Denote by $w'_\pi(p)$ the directional derivative of $w$ at $p$ in the direction of $\pi$ (which exists by Alexandrov's theorem).
            Concavity implies that for all $q$ in a $\pi$-neighborhood of $p$, $w(q) \leq w(p) + w'_\pi(p) (q - p)$ but since $w$ is strictly concave in no $\pi$-neighborhood of $p$ we must be able to find a $q$ such that this holds with equality. 
            Fix such a $q$ and now consider the belief process which stays at $p$ until a random time when it jumps to $q$, and that random time is given by:
            \begin{align*}
            T \sim \mathcal{E} \biggl(\lambda \frac{\pi-p}{q-p} \biggr)
            \end{align*}
            Denote $\rho:= \lambda \frac{\pi-p}{q-p}$. 
            It is direct to verify that this is a feasible belief process, assuming any arbitrary consistent distribution following the jump to $q$. 
            First decompose the expectation of $P_t$ conditionally on the jump time, replace with explicit expression depending on whether or not the jump has occurred at $t$, then rearrange and compute integrals to very that $\Esp[P_t] = p + \bigl( 1- e^{-\lambda t} \bigr)\bigl( \pi - p \bigr)$, i.e. $P$ is feasible.
            To show that $P$ is optimal, it suffices to establish that the expected payoff at $p$ attains $w(p)$.
            First write explicitly:
            \begin{align*}
                \Esp \Biggl[ \int_0^T e^{-r t} f(p) dt + e^{-rT} w(q) \Biggr] & = \Esp \biggl[ 1- e^{-rT} \biggr] \frac{f(p)}{r}+ \Esp \biggl[ e^{-rT} \biggr] w(q)
                \\
                & =\Biggl( 1 -\frac{\rho}{r+ \rho} \Biggr)  \frac{f(q)}{r} + \frac{\rho}{r+ \rho} w(q).
            \end{align*}
            Now observe that by optimality since waiting is always feasible at any point it must be that $r w(p) \geq f(p) + \lambda (\pi-p) w'_\pi(p)$. 
            Given the assumptions on $p$ and the previous arguments, it must be optimal to wait arbitrary close to $p$ in the opposite direction from $\pi$, hence this actually must hold with equality at $p$. 
            Furthermore recall that $q$ has been chosen so that $w(q) = w(p) + w'_\pi(p) (q - p)$.
            Replacing in the previous expression and rearranging gives that the payoff at $p$ is equal to:
            \begin{align*}
                \frac{r}{r+\rho} & \biggl( w(p)   - \frac{\lambda}{r} (\pi-p) w'_\pi(p) \biggr) + \frac{\rho}{r+\rho} \biggl(  w(p) + w'_\pi(p) (q - p) \biggr)
                \\
                & = w(p) + w_\pi'(p) \Biggl( \frac{\rho}{r+\rho} (q - p) - \frac{\lambda}{r} \frac{r}{r+\rho}  (\pi-p) \Biggr)
                \\
                & = w(p) + w_\pi'(p) \Biggl( \frac{\lambda \frac{\pi-p}{q-p}}{r+\rho} (q - p) - \frac{\lambda}{r} \frac{r}{r+\rho}  (\pi-p) \Biggr)
                \\
                & = w(p) + w_\pi'(p) \Biggl( \frac{\lambda (\pi-p)}{r+\rho} - \frac{\lambda (\pi-p) }{r+\rho}  \Biggr) = w(p)
            \end{align*}
            Which concludes the proof.   
        \end{proof}
        \begin{appendixproposition}
        \label{prop:convergence-to-jump-policy}
            Let $P^\kappa$ the optimal belief process for $\kappa>0$ and let $\inforegion_\kappa$ the corresponding information acquisition region. If $P^\kappa$ converges in distribution to $P$, then $P$ is a wait-or-confirm belief process with initial jump beliefs region $\liminf_{\kappa \downarrow 0} \inforegion_\kappa$
        \end{appendixproposition}
        \begin{proof}
            Consider the optimal belief process $P^\kappa$ for $\kappa>0$ and denote its information acquisition region $\inforegion_\kappa$. \
        Assume $P^\kappa$ converges to $P$ in distribution (i.e in $\beliefprocesses$ equipped with its weak-* topology induced by the Skorohod metric) as $\kappa$ goes to zero.
        First consider a belief $p \in \liminf_{\kappa \rightarrow 0} \inforegion_\kappa$, i.e $p$ is evenutally in all information acquisition regions for $\kappa$ small enough.
        Because convergence in distribution implies convergence in distribution at all continuity points and $P$ is càdlàg, hence in particular continuous at the initial time, this must imply that $P$ involves immediate information acquisition at $p$.
        (In other words, $P_0=p$ would involve a contradiction, so the initial distribution of $P_0$ must involve immediate information acquisition.) 

        Now consider instead a belief $p \in \Biggl( \liminf_{\kappa \rightarrow 0} \inforegion_\kappa \Biggr)^c \cup \outboundary \liminf_{\kappa \rightarrow 0} \inforegion_\kappa = \piint \liminf_{\kappa \rightarrow 0} \inforegion_\kappa^c$ i.e there exist a $\pi$-neighborhood $b_\pi(p,\varepsilon)$ of $p$ such that all points in $b_\pi(p,\varepsilon)$ are eventually not in all information acquisition regions for $\kappa$ small enough.
        Again using convergence in distribution at all continuity points, it must be that the limit process involves no information acquisition arbitrarily close to $p$ in the direction of $\pi$, which in light of the arguments in the proof of \hyref{Proposition}{prop:optimal-jump-policy-from-w0} establishes that waiting must be uniquely optimal in a $\pi$-neighborhood of $p$ -- hence by optimality of the limit $P$ must involve waiting at $p$.

        Lastly, consider $p$ is neither of the previous sets.
        Without loss assume $p<\pi$ (the other case is symmetrical).
        From the previous points, this must mean that $p$ is arbitrarily close in the direction of $\pi$ to a point where immediate information acquisition is optimal in the limit problem, and arbitrarily close in the opposition direction from a point where waiting is uniquely optimal in the limit problem.  
        By definition since $p \notin \liminf_{\kappa \rightarrow 0} \inforegion_\kappa$ this means we can find a sequence $\kappa_n$ converging to zero such that for all $n$, $p \in \inforegion^c_{\kappa_n}$ i.e no information acquisition occurs at $p$ under $\kappa_n$. 
        Let $z_n$ the closest point to $p$ in the direction of $\pi$ which is in $\inforegion_{\kappa_n}$. 
        Again by definition, it must be that $z_n$ gets arbitrarily close to $p$ as $n$ goes to infinity otherwise this would contradict $p \notin \piint \liminf_{\kappa \rightarrow 0} \inforegion_\kappa^c$, so $z_n \rightarrow p$.
        Denote $\{q^0_n,q^1_n\}$ the support of the optimal experiment at $z_n$ for any $n$.
        By construction $p \in [q^0_n,q^1_n]$ for all $n$. 
        Up to a subsequence, denote $q^0_\infty,q^1_\infty$ the limits of $q^0_n,q^1_n$ respectively; clearly $p \in [q^0_\infty,q^1_\infty]$.
        First consider the possibility that $q^0_\infty<q^1_\infty$ and $p \in (q^0_\infty,q^1_\infty)$.
        Then observe that for all $n$, for any $q \in [q^0_\infty,q^1_\infty]$:
        \begin{align*}
        \Cav[w_{\kappa_n}](q) = \frac{q-q^0_n}{q^1_n-q^0_n} w_{\kappa_n}(q^1_n) + \frac{q^1_n-q}{q^1_n-q^0_n} w_{\kappa_n}(q^0_n)
        \end{align*}
        Continuity of $w_\kappa(q)$ in $q$ and pointwise convergence in $\kappa$, along with the assumption that $q^0_\infty<q^1_\infty$ implies that $w_0$ is locally affine at $p$, which contradicts the fact that $p$ is arbitrarily close in the opposition direction from $\pi$ to a point where waiting is uniquely optimal in the limit problem.
        Now consider instead the possibilty that $p = q^1_\infty$: again this implies a contradiction because by assumption on $p$, $w_0$ must be strictly concave in some $\pi$-neighborhood of $p$.
        Hence the only remaining possibility is $p=q^0_\infty$ and $q^0_\infty < q^1_\infty$.

        Having established that both $q_0^n$ and $z_n$ converge to $p$, and that $p < q^1_\infty$, it remains to establish that the distribution of the belief process at $p$ converges to the "confirmation" process which stays at $p$ until an exponentially distributed jump time to $q^1_\infty$.
        To do so, it suffices to consider the distribution of the time $T_n$ that it takes for the process to reach $q^1_n$ from $p$.
        This can be expressed as:
        \begin{align*}
            T_n = \tau(p,z_n) + \tau(q^0_n,z_n) X_n \text{ where } X_n \sim \mathcal{G} \biggl( \frac{z_n-q^0_n}{q^1_n-q_0^n} \biggr)
        \end{align*}
        Where $\mathcal{G}$ denotes the geometric distribution as before. 
        From the previous arguments, the first term $\tau(p,z_n)$ goes to zero as $n$ converges to infinity, hence it is enough to prove that:
        \begin{align*}
            \tau(q^0_n,z_n) \times X_n \xrightarrow[n\rightarrow \infty]{d} \mathcal{E} \biggl( \lambda \frac{\pi - p}{q^1_\infty-p} \biggr).
        \end{align*}
        Which is established by a direct computation on the CDFs of $X_n$.   
    \end{proof}

        \paragraph{Proof of \hyref{Theorem}{thm:explicit-stationary-policy-no-fixed-cost}}
        The proof relies first on establishing the upper bound for $w_0$, then on exhibiting a policy which achieves it in some region around $\pi$.
        The upper bound is straightforwardly derived from applying Jensen's inequality pathwise and pointwise:
        \begin{align*}
            w_0(p) = \sup_{P \in \beliefprocesses(p)} \Esp \Biggl[ \int_0^\infty e^{-rt} f(P_t) dt \Biggr] 
            & \leq \sup_{P \in \beliefprocesses(p)} \Esp \Biggl[ \int_0^\infty e^{-rt} \Cav[f](P_t) dt \Biggr] 
            \\
            & \leq \sup_{P \in \beliefprocesses(p)} \int_0^\infty e^{-rt} \Cav[f] \bigl( \Esp[P_t] \bigr) dt
            = \int_0^\infty e^{-rt} \Cav[f] \bigl( p_t \bigr) dt
        \end{align*}
        Consider first the case where $\Cav[f] = f$ in a neighborhood of $\pi$. 
        In that case, not acquiring information at any $p$ in this neighborhood clearly achieves the upper bound, uniquely so if $f$ is locally strictly concave (by another application of Jensen's inequality).
        This must mean that no information is acquired in the long run under any optimal belief process in the $\kappa=0$ problem.
        In the case where $\Cav[f](\pi)=f(\pi)$ but not $\Cav[f] = f$ in a neighborhood around $\pi$, some straightforward but tedious casework shows that either (i) jumping to $\pi$ and then not acquiring information (when approaching from a side where $\Cav[f]$ is linear) or (ii) eventually stopping information acquisition must be optimal in a neighborhood of $p$ (when approaching from a side where $\Cav[f]$ is strictly concave). 
        In either case, uniqueness cannot be guaranteed because of knife-edge indifferences if $f$ is locally affine.

        The main part of the proof consists of establishing optimality of the "confirmatory" policy when $\Cav[f](\pi) > \pi$. 
        In that case, denote $(q_0,q_1)$ an interval such that $\pi \in (q_0,q_1)$, $\Cav[f]>f$ in $(q_0,q_1)$ and $\Cav[f]=f$ at $q_0$ and $q_1$.
        Fix any initial belief $p \in [q_0,q_1]$ and consider the belief process which immediately jumps to $\{q_0,q_1\}$ if $p \in (q_0,q_1)$; jumps from $q_0$ to $q_1$ at rate $\rho_0 := \lambda \frac{\pi-q_0}{q_1-q_0}$ and from $q_1$ jumps to $q_0$ at rate $\rho_1 := \lambda \frac{q_1-\pi}{q_1-q_0}$.
        Denote $Q_t$ the resulting Markov chain, $\Psi$ its rate matrix and $M(t)$ the matrix of conditional probabilities.
        $M(t)$ solves the Kolmogorov equation $M'(t) = M(t) \Psi$ i.e $M(t) = e^{t \Psi}$ and in this case simplifies to an explicit expression:
        \begin{align*}
        M(t) := \begin{pmatrix} \prob(Q_t = 0 | Q_0 = q_0) & \prob(Q_t = 1 | Q_0 = q_0) \\ \prob(Q_t = 0 | Q_0 = q_1) & \prob(Q_t = 1 | Q_0 = q_1) \end{pmatrix} = \begin{pmatrix} 
                \frac{q_1-\pi}{q_1-q_0} + \frac{\pi-q_0}{q_1-q_0} e^{-\lambda t} & \frac{\pi-q_0}{q_1-q_0} - \frac{\pi-q_0}{q_1-q_0} e^{-\lambda t}  
                \\ 
                \frac{q_1-\pi}{q_1-q_0} - \frac{q_1-\pi}{q_1-q_0} e^{-\lambda t}  & \frac{\pi-q_0}{q_1-q_0} + \frac{q_1-\pi}{q_1-q_0} e^{-\lambda t}  
            \end{pmatrix}.
        \end{align*}
        This, in particular, allows to verify that this belief process is feasible since $\Esp[Q_0]=p$ by construction and we can directly compute from the explicit expression of $M(t)$, skipping algebraic simplifications:
        \begin{align*}
            \Esp[Q_t|Q_0=q_0] = \pi - (\pi - q_0) e^{-\lambda t} \text{ and } \Esp[Q_t|Q_0=q_1] = \pi + (q_1-\pi) e^{-\lambda t}.
        \end{align*}
        To compute the induced expected value observe that for anyt $t$, since $Q_t \in \{q_0,q_1\}$ a.s. then $f$ and $\Cav[f]$ coincide over the support of $Q_t$, hence $\Esp[f(Q_t)] = \Esp \bigl[ \Cav[f](Q_t) \bigr]$.
        Furthermore since $\Cav[f]$ is affine over $[q_0,q_1]$ and given the compensated martingale constraint: $\Esp \bigl[ \Cav[f](Q_t) \bigr] = \Cav[f] \bigl( \Esp[Q_t] \bigr) = \Cav[f](p_t)$.
        This immediately means that:
        \begin{align*}
            \Esp \Biggl[ \int_0^\infty e^{-rt} f(Q_t) dt \Biggr] = \int_0^\infty e^{-rt} \Cav[f](p_t) dt,
        \end{align*}
        which proves the desired result.
    
\end{document}